\documentclass[11pt]{article}
\usepackage[bookmarksopen=true,bookmarksnumbered=true,bookmarksopenlevel=4,CJKbookmarks=true]{hyperref}
\usepackage{bookmark}
\usepackage{multirow}
\usepackage{amsthm,amssymb,amsmath,mathrsfs,cite,cleveref,float,setspace}
\usepackage{enumerate}
\usepackage[shortlabels]{enumitem}
\usepackage[top=0.8 in, left=0.9 in, right=0.9 in, bottom=0.8 in]{geometry}
\usepackage{pdflscape}

\usepackage{xargs}
\usepackage[colorinlistoftodos]{todonotes}
\newcommandx{\hcomment}[2][1=]{\todo[linecolor=blue,backgroundcolor=blue!25,bordercolor=blue,#1]{#2}}

\usepackage{graphicx,caption,subcaption,epstopdf,enumitem,tikz,standalone}%stackengine
\usetikzlibrary{decorations,shapes,backgrounds}
\usetikzlibrary{positioning,decorations.pathmorphing,decorations.pathreplacing}

\newtheorem {theorem}{Theorem}
\newtheorem {lemma}{Lemma}
\newtheorem {proposition}{Proposition}
\newtheorem {corollary}{Corollary}
\newtheorem {definition}{Definition}
\theoremstyle{remark}
\newtheorem*{remark}{Remark}
\usepackage[noend]{algpseudocode}
\usepackage{algorithm,algorithmicx}

\graphicspath{ {figs/} }

\usepackage{array}
\newcolumntype{P}[1]{>{\centering\arraybackslash}p{#1}}
\newcolumntype{M}[1]{>{\centering\arraybackslash}m{#1}}

\newcommand{\calH}{\mathcal{H}}
\newcommand{\calT}{\mathcal{T}}

\newcommand{\calF}{\mathcal{F}}

\newcommand{\calM}{\mathcal{M}}
\newcommand{\up}{\mathbb{U}}
\newcommand{\horiz}{\mathbb{H}}
\newcommand{\down}{\mathbb{D}}
\newcommand{\whp}{\alpha} % greek letter we use to refer to the "weakly-helly" parameter
\newcommand{\hg}{Helly-gap}
\mathchardef\mhyphen="2D
\newcommand{\ecc}{e}

\date{\today}

\title{Helly-gap of a graph and vertex eccentricities}
\author{\medskip Feodor F. Dragan and Heather M. Guarnera   \\%} 
{\small Algorithmic Research Laboratory, Department of Computer Science} \\
{\small Kent State University, Kent, Ohio, USA} \\
\href{mailto:dragan@cs.kent.edu}
{\small\texttt{dragan@cs.kent.edu}},
\href{mailto:hmichaud@kent.edu}
{\small\texttt{hmichaud@kent.edu}}
}

\begin{document}

\maketitle

\begin{abstract} %[I still need to work on this]
A new metric parameter for a graph, Helly-gap, is  introduced. A graph~$G$ is called $\whp$-weakly-Helly if any system of pairwise intersecting disks in~$G$ has a nonempty common  intersection when the radius of each disk is increased by an additive value $\whp$. The minimum $\whp$ for which a graph $G$ is $\whp$-weakly-Helly is called the \hg{} of~$G$ and denoted by $\whp(G)$. 
The \hg{} of a graph~$G$ is characterized by distances in the injective hull $\calH(G)$, which is a (unique) minimal Helly graph which contains~$G$ as an isometric subgraph. This characterization is used as a tool to generalize many eccentricity related results known for Helly graphs ($\whp(G)=0$), as well as for chordal graphs ($\whp(G)\le 1$), distance-hereditary graphs ($\whp(G)\le 1$) and $\delta$-hyperbolic graphs ($\whp(G)\le 2\delta$), to all %$\whp$-weakly-Helly
graphs, parameterized by their Helly-gap $\whp(G)$. 
Several additional graph classes are shown to have a bounded \hg{}, including AT-free graphs and graphs with bounded tree-length, bounded chordality or bounded $\alpha_i$-metric. 
%, to name a few. 
\medskip

\end{abstract}

\section{Introduction}
All graphs in this paper are unweighted, undirected, connected, and without loops or multiple edges. 
In a graph $G=(V(G),E(G))$, a disk $D_G(v,r)$ with radius $r$ and centered at a vertex~$v$ consists of all vertices with distance at most~$r$ from~$v$, i.e., $D_G(v,r) = \{u \in V(G) : d_G(u,v) \le r\}$.
A graph $G$ is called \emph{Helly} if every system of pairwise intersecting disks of $G$ has a non-empty intersection.

\begin{definition}[Helly graph]
A graph $G$ is Helly if, for any system of disks $\calF = \{ D_G(v, r(v)) : v \in S \subseteq V(G)\}$, the following Helly property holds:
if $X \cap Y \neq \emptyset$ for every $X,Y \in \calF$,
then $\underset{v \in S}{\bigcap} D_G(v, r(v)) \neq \emptyset$.
\end{definition}

Helly graphs are well investigated. They have several characterizations and important  features as established in~\cite{FDraganPhD,DraganCenters,Bandelt:1989:DAR:72175.72177,NOWAKOWSKI1983223,QUILLIOT1985186,BANDELT199134}.  %which we briefly summarize. 
They are exactly the so-called {\em absolute retracts of reflexive graphs} and
possess a certain elimination scheme~\cite{FDraganPhD,DraganCenters,Bandelt:1989:DAR:72175.72177,NOWAKOWSKI1983223,BANDELT199134}. They can be recognized in $O(n^2m)$ time~\cite{FDraganPhD} where $n$ is the number of vertices and $m$ is the number of edges. The Helly property works as a compactness criterion on graphs~\cite{QUILLIOT1985186}. More importantly, every graph is isometrically embeddable into a Helly graph~\cite{Isbell1964/65,DRESS1984321,ursLang} (see Section \ref{inj-hull} for more details). 

%
% unused
%They are characterized by a Helly-extremal vertex ordering which is also used for recognition in $O(n^2m)$ time.
Many nice properties of Helly graphs are based on the eccentricity $\ecc_G(v)$ of a vertex~$v$, which is defined as the maximum distance from~$v$ to any other vertex of the graph (i.e., $\ecc_G(v) = \max_{u \in V(G)}d_G(v,u)$).
The minimum (maximum) eccentricity in a graph $G$ is the radius $rad(G)$ (diameter $diam(G)$).
The center $C(G)$ of a graph $G$ is the set of all vertices with minimum eccentricity (i.e., $C(G) = \{v \in V(G) : \ecc_G(v) = rad(G)\}$.
Graph's diameter is tightly bounded in Helly graphs as $2rad(G) \ge diam(G) \ge 2rad(G) - 1$~\cite{FDraganPhD,DraganCenters}.
Moreover, the eccentricity function in Helly graphs is unimodal~\cite{DraganCenters}, that is, any local minimum coincides with the global minimum;
this is equivalent to the condition that, for any vertex $v \in V(G)$, $\ecc_G(v) = d_G(v,C(G)) + rad(G)$ holds~\cite{FDraganPhD}. The unimodality of the eccentricity function was recently used in \cite{DBLP:journals/corr/abs-1910-10412,DucoffeNEW} to compute the radius, diameter and a central vertex  of a Helly graph in subquadratic time. For a vertex $v \in V(G)$, let  $F_G(v)$ denote 
the set of all vertices farthest from $v$, that is, $F_G(v) = \{ u \in V(G) : \ecc_G(v) = d_G(v,u) \}$. 
For every vertex $v$ of a Helly graph $G$, each vertex $u \in F_G(v)$ satisfies  $\ecc_G(u) \ge 2rad(G) - diam(C(G))$~\cite{FDraganPhD}. 
Although the  center $C(G)$ of a Helly graph $G$ may have an arbitrarily large diameter (as any Helly graph is the center of some Helly graph), $C(G)$ induces a Helly graph and is isometric in $G$~\cite{FDraganPhD}.
Additionally, any power of a Helly graph is a Helly graph as well~\cite{FDraganPhD}.

In this paper, we introduce a far reaching generalization of Helly graphs, $\whp$-weakly-Helly graphs. We define \emph{$\whp$-weakly-Helly graphs} as those graphs which are ``weakly'' Helly with the following simple generalization:
for any system of disks, if the disks pairwise intersect, then by expanding the radius of each disk by some integer $\whp$ there forms a common intersection.
Thus, Helly graphs are exactly 0-weakly-Helly graphs.

\begin{definition}[$\whp$-weakly-Helly graph]
A graph $G$ is $\whp$-weakly-Helly if, for any system of disks $\calF = \{ D_G(v, r(v)) : v \in S \subseteq V(G)\}$, the following $\whp$-weakly-Helly property holds:
if $X \cap Y \neq \emptyset$ for every $X,Y \in \calF$,
then $\underset{v \in S}{\bigcap} D_G(v, r(v) + \whp) \neq \emptyset$.\\  Clearly, every graph is $\whp$-weakly-Helly for some $\alpha$. We call the minimum~$\whp$ for which a graph~$G$ is $\whp$-weakly-Helly the \hg{} of~$G$, denoted by $\whp(G)$.
\end{definition}

Interestingly, there are a few results in the literature which demonstrate that such a $\alpha$-weakly-Helly property with bounded $\alpha$ is present in some graphs and in some metric spaces.  
In~\cite{DBLP:conf/approx/ChepoiE07}, Chepoi and  Estellon showed that for every $\delta$-hyperbolic geodesic space $(X, d)$ (and for every $\delta$-hyperbolic graph $G$), if disks 
of the family $\calF = \{ D(x, r(x)) : v \in S\subseteq X, S$ is compact$\}$ pairwise intersect, then the disks $\{D(x, r(x) + 2\delta) : x \in S\}$ have a nonempty common intersection (see also ~\cite{Chepoi_2008}). That is, the disks in $\delta$-hyperbolic geodesic spaces and in $\delta$-hyperbolic graphs satisfy $(2\delta)$-weakly-Helly property. Even earlier, Lenhart, Pollack et al.~\cite[Lemma 9]{link-center} established that the disks of simple polygons endowed with the link distance satisfy a Helly-type property which implies 1-weakly-Helly 
property. For chordal graphs~\cite{DBLP:journals/dm/DraganB96} and for distance-hereditary graphs~\cite{10.1007/3-540-58218-5_34} (as well as for a more general class of graphs~\cite{DBLP:journals/dm/Chepoi98}) the following similar result is  known: if for a set $M\subseteq V(G)$ and radius function $r: M\rightarrow \mathbb{N}$, every two vertices $x,y\in M$ satisfy $d_G(x,y)\le r(x)+r(y)+1$, then there is a clique $K$ in $G$ such that $d_G(v,K)\le r(v)$ for all $v\in M$. Hence, clearly,  every chordal graph and every distance-hereditary graph is 1-weakly-Helly. Note that two disks $D_G(v,p)$ and $D_G(u,q)$ intersect if and only if $d_G(v,u)\le p+q$. 

There are numerous other approaches to generalize in some form the Helly graphs. One may loosen the restriction on the type of sets which must satisfy the Helly property. If one considers neighborhoods or cliques, rather than arbitrary disks, then one gets the neighborhood-Helly graphs and the clique-Helly graphs as superclasses of Helly graphs (see \cite{DBLP:journals/ipl/LinS07,DBLP:journals/dam/GroshausLS17,doi:10.1137/1.9780898719796} and papers cited therein). 
One may also generalize the Helly-property that a family of sets satisfies. % there are too many
This can be accomplished by specifying a minimum number or size of sets that pairwise intersect, a minimum number of subfamilies that have a common intersection, or the size of the intersection (see \cite{ECKHOFF1993389} and papers cited therein). Far reaching examples include the $(p,q)$-Helly %hypergraphs 
property (see~\cite{ALON1992103,ALON200279} and papers cited therein) and %hypergraphs which satisfy 
the fractional Helly property  (see~\cite{10.2307/2042758,BARANY2003227} and papers cited therein). 
%In the $(p,q)$-disk-Helly graphs~\cite{tuzaHelly,voloshin}, any family pairwise intersecting disks of size~$p$ or fewer have a common intersection of size at least~$q$.
%
 There is also a related to ours notion of $\lambda$-hyperconvexity in metric spaces $(X,d)$~\cite{espinola2001,khamsi2001}, where $\lambda$ is the smallest multiplicative constant~such that for any system of disks $\calF = \{ D(x, r(x)) : x \in S \subseteq X\}$, the following property holds: if $X \cap Y \neq \emptyset$ for every $X,Y \in \calF$, then $\underset{x \in S}{\bigcap} D(x, \lambda \cdot r(x)) \neq \emptyset$. 
%Clearly, if $\whp(G)$ is bounded, then~$G$ is $\lambda$-hyperconvex for bounded~$\lambda$.
Several classical metric spaces are $\lambda$-hyperconvex for a bounded $\lambda$ (see~\cite{espinola2001} and papers cited therein). These include reflexive Banach spaces and dual Banach spaces ($\lambda \le 2$) and Hilbert spaces ($\lambda \le \sqrt 2$). %\hcomment{\small is OK? want to mention examples, but our paper isn't metric space or hypergraph based}
The classical Jung Theorem asserting that each subset $S$ of the Euclidean space $\mathbb{E}^m$ with finite diameter $D$ is contained in a disk of radius at most $\sqrt{\frac{m}{2(m+1)}}D$ belongs to this kind of results, too.  

%\centerline{.......} 

In this paper, we are interested in $\whp$-weakly-Helly graphs, where $\whp$ describes an additive constant by which radius of each disk in the family of pairwise intersecting disks can be 
%expanded 
increased in order to obtain a nonempty common intersection of all expanded disks.
By definition, any $\whp$-weakly-Helly graph is $(\whp+1)$-hyperconvex. 
%
%we are interested in additive expantion ... 
We find that there are also an abundance of graph classes  with a bounded \hg{} $\whp(G)$.
Here, 
we generalize many eccentricity related results known for Helly graphs (as well as for chordal graphs, distance-hereditary graphs and $\delta$-hyperbolic graphs) to all %$\whp$-weakly-Helly
graphs, parameterized by their Helly-gap $\whp(G)$.
%
%In this paper, we generalize the eccentricity related results known for Helly graphs to much larger family of graphs, $\whp$-weakly-Helly graphs. 
We provide a characterization of weakly-Helly graphs with respect to their injective hulls and use this as a tool to prove those generalizations. Several additional well-known graph classes that are $\whp$-weakly-Helly for some constant $\whp$ are also identified. 

The main results of this paper can be summarized as follows. After providing in Section~\ref{sec:prelim} necessary definitions, notations and some  relevant results on Helly graphs, in Section~\ref{inj-hull} we give a characterization of $\whp$-weakly-Helly graphs through distances in injective hulls. The \emph{injective hull} of a graph $G$, denoted by $\calH(G)$, is a (unique) minimal Helly graph which contains $G$ as an isometric subgraph~\cite{Isbell1964/65,DRESS1984321}. We show that $G$ is an $\whp$-weakly-Helly graph if and only if for every vertex $h \in V(\calH(G))$ there is a vertex $v \in V(G)$ such that $d_{\calH(G)}(h,v) \leq \whp$ (Theorem~\ref{thm:closeRealVertex}). In Section~\ref{sec:diam-rad-ecc}, we relate the diameter, radius, and all eccentricities in~$G$ to their counterparts in~$H:=\calH(G)$. In particular, we show that $\ecc_G(v) = \ecc_H(v)$ for all $v\in V(G)$, $diam(G) = diam(H)$, and  $rad(G) - \whp(G) \le rad(H) \le rad(G)$. Additionally, we show  
$diam_G(C(G)) -  2\whp(G) \le diam_H(C(H)) \le diam_G(C^{\whp(G)}(G)) + 2\whp(G)$, where $C^\ell(G) = \{ v \in V(G) : \ecc_G(v) \le rad(G) + \ell\}$.
We also provide several bounds on $\whp(G)$ including its relation to diameter and radius as well as investigate the \hg{} of powers of weakly-Helly graphs (Theorem~\ref{thm:hGapBounds}).
In particular, $\lfloor (2rad(G) - diam(G))/2 \rfloor \le \whp(G) \le \lfloor diam(G)/2 \rfloor$ holds.
The eccentricity function in $\whp$-weakly-Helly graphs is shown to exhibit the property that any vertex~$v \notin C^\whp(G)$ has a nearby vertex, within distance~$2\whp+1$ from~$v$, with strictly smaller eccentricity (Theorem~\ref{thm:unimodality}).
In Section~\ref{sec:estimatingEcc}, we give upper and lower bounds on the eccentricity~$\ecc_G(v)$ of a vertex~$v$.
We consider bounds based on the distance from~$v$ to a closest vertex in $C^{\alpha(G)}(G)$% with eccentricity at most $rad(G)+\whp(G)$
, whether~$v$ is farthest from some other vertex and if $diam_G(C^{\whp(G)}(G))$ is bounded.
In particular, we show that 
%$d_G(v,C^{\whp(G)}(G)) + rad(G) - \whp(G) \le 
$|\ecc_G(v) - d_G(v,C^{\whp(G)}(G)) - rad(G)| \le \whp(G)$ holds for any vertex~$v \in V(G)$ (Theorem~\ref{thm:eccentricities}).
We also prove the existence of a spanning tree~$T$ of $G$ which gives an approximation of all vertex eccentricities in~$G$ with an additive error depending only on $\whp(G)$ and $diam_G(C^{\whp(G)}(G))$ (Theorem~\ref{thm:spanningTree}).
We find that in any shortest path to $C^{\whp(G)}(G)$, the number of vertices with locality more than 1 does not exceed $2\whp(G)$, where the locality of a vertex~$v$ is the minimum distance from~$v$ to a vertex of strictly smaller eccentricity (Theorem~\ref{thm:upHorizontalEdgesBoundWH}). All these results greatly generalize some known facts about distance-hereditary graphs, chordal graphs, and $\delta$-hyperbolic graphs. Those graphs have bounded Helly gap. 
In Section~\ref{sec:cases}, we identify several more  (well-known) graph classes with a bounded \hg{}, including $k$-chordal graphs, AT-free graphs, rectilinear grids, graphs with a bounded $\alpha_i$-metric, and graphs with bounded tree-length or tree-breadth. We conclude with open questions in Section~\ref{sec:conclusion}. 

%%%%%%%%%%%%%%%%%%%%%%%%%%%%%%%%%%%%%%%%%%%%%%%%%%
% CONTRIBUTIONS OF THIS PAPER
% - Define weakly-Helly. Characterize wrt injective hull
% - Give other Helly-like properties: 
%   - (?????) Vertex ordering like h-extremal
%   - (?????) Recognition
%   - (?????) Powers of weakly-helly graphs (are weakly helly? are helly?)
%   - Bound on diam(G)
%   - Almost unimodality as described by:
%           - ecc inequality
%           - ecc terrain (up & horizontal edges)
%   - (?????) Centers
%%%%%%%%%%%%%%%%%%%%%%%%%%%%%%%%%%%%%%%%%%%%%%%%%%

\section{Preliminaries}\label{sec:prelim}
All graphs $G=(V(G),E(G))$ occurring in this paper are undirected, connected, and without loops or multiple edges.
The \emph{length} of a path from a vertex~$u$ to a vertex~$v$ is the number of edges in the path.
The \emph{distance} $d_G(u,v)$ between two vertices~$u$ and~$v$ is the length of a shortest path connecting them in $G$.
The distance from a vertex~$u$ to a vertex set $S \subseteq V(G)$ is defined as $d_G(u,S) = \min_{s \in S}d_G(u,s)$. 
The \emph{eccentricity} of a vertex is defined as $\ecc_G(v) = \max_{u \in V(G)}d_G(v,u)$.
The vertex set $F_G(v)$ for a vertex $v \in V(G)$ denotes the set of all vertices farthest from $v$, that is, $F_G(v) = \{ u \in V(G) : \ecc_G(v) = d_G(v,u) \}$.
A \emph{disk} of radius~$k$ centered at a set~$S$ (or a vertex) is the set of vertices of distance at most~$k$ from~$S$, that is, $D_G(S,k) = \{u \in V(G): d_G(u,S) \leq k\}$. %The \emph{(closed) neighborhood} $N[v]$ is $D(v,1)$.
We omit~$G$ from subindices when the graph is known by context.

The $k^{th}$ power $G^k$ of a graph $G$ is a graph that has the same set of vertices, but in which two distinct vertices are adjacent if and only if their distance in $G$ is at most $k$. A subgraph $G'$ of a graph $G$ is called {\em isometric} (or a {\em distance preserving subgraph}) if for any two vertices $x,y$ of $G'$, $d_G(x,y)=d_{G'}(x,y)$ holds.
The \emph{interval} $I(x,y)$ is the set of all vertices belonging to a shortest $(x,y)$-path, that is, $I(x,y) = \{u \in V(G) : d(x,u) + d(u,y) = d(x,y)\}$.
The interval \emph{slice} $S_k(x,y)$ is the set of vertices on $I(x,y)$ which are at distance~$k$ from vertex~$x$.
An interval $I(x,y)$ is said to be $\kappa$-thin if, for any $0 \le \ell \le d(x,y)$ and any $u,v \in S_\ell(x,y)$, $d(u,v) \le \kappa$ holds. The smallest $\kappa$ for which all intervals of $G$ are $\kappa$-thin is called the \emph{interval thinness} of $G$ and denoted by $\kappa(G)$.
By $G - \{x\}$ we denote an induced subgraph of~$G$ obtained from~$G$ by removing a vertex~$x \in V(G)$. 

The minimum (maximum) eccentricity in a graph $G$ is the \emph{radius} $rad(G)$ (\emph{diameter} $diam(G)$).
The \emph{center} $C(G)$ is the set of vertices with minimum eccentricity (i.e., $C(G) = \{v \in V(G) : \ecc(v) = rad(G)\}$.
It will be useful to define also $C^k(G) = \{v \in V(G) : \ecc(v) \le rad(G) + k\}$; then $C(G) = C^0(G)$.
Let $M \subseteq V(G)$ be a subset of vertices of~$G$.
We distinguish the eccentricity with respect to~$M$ as follows: denote by $\ecc_G^M(v)$ the maximum distance from vertex~$v$ to any vertex $u \in M$, i.e., $\ecc_G^M(v) = \max_{u \in M} d_G(v,u)$.
We then define
$F_G^M(v) = \{ u \in M : \ecc_G^M(v) = d_G(v,u) \}$,
$rad_G(M) = \min_{v \in V(G)} \ecc_G^M(v)$, and
$diam_G(M) = \max_{v \in M} \ecc_G^M(v)$.
Let $C_G^\ell(M) = \{ v \in V(G) : \ecc_G^M(v) \le rad_G(M) + \ell\}$ and $C_G(M) = C_G^0(M)$.
For simplicity, when $M=V(G)$, we continue to use the notation $rad(G)$, $C(G)$, and $diam(G)$.

The following results for Helly graphs will prove useful for $\whp$-weakly-Helly graphs.
\begin{lemma}{\bf\cite{FDraganPhD}}\label{lem:allCentersHellyIsom}
Let~$G$ be a Helly graph. For every $M \subseteq V(G)$, the graph induced by the center $C_G(M)$ is Helly and is an isometric subgraph of~$G$.
\end{lemma}

Given this lemma, it will be convenient to also denote by $C_G(M)$ a subgraph of~$G$ induced by $C_G(M)$.

\begin{lemma}{\bf\cite{FDraganPhD,Bandelt:1989:DAR:72175.72177}}\label{lem:removeDominatingVertex}
Let $G$ be a Helly graph. If there are two distinct vertices~$w,x \in V(G)$ such that $D(w,1) \supseteq D(x,1)$, then $G - \{x\}$ is Helly and an isometric subgraph of~$G$.
\end{lemma}
%\begin{proof}
%As $D(x,1) \subseteq D(w,1)$, any vertex~$u \in V(G) \setminus \{x\}$ satisfies $d_G(u,w) \le d_G(u,x)$. Hence, if $x \in I(u,y)$ for two vertices~$u,y \in V(G) \setminus \{x\}$, then $w \in I(u,y)$. Thus, $G - \{x\}$ is isometric.
%
%Consider any family~$\calF$ of pairwise intersecting disks in the isometric subgraph~$G - \{x\}$. $\calF$ is also pairwise intersecting in~$G$.
%Since~$G$ is Helly, there is a vertex~$z \in V(G)$ common to all disks of~$\calF$.
%Let $z = x$ (otherwise $z \in V(G - \{x\})$ and we are done).
%Then, any disk $D_{G - \{x\}}(v,r) \in \calF$ satisfies $d_{G-\{x\}}(v,w) = d_G(v,w) \le d_G(v,x) \le r$.
%The Helly property is satisfied as vertex~$w$ is common to all disks from $\calF$  in $G - \{x\}$.
%\end{proof}

For a graph $G$ and a subset $M\subseteq V(G)$, the eccentricity function $e_G^M(\cdot)$ is called {\em unimodal} if every vertex $v \in V(G) \setminus C_G(M)$ has a
neighbor~$u$ such that $\ecc_G^M(u) < \ecc_G^M(v)$. 

\begin{lemma}{\bf\cite{FDraganPhD,DraganCenters}}\label{lem:unimodalHelly}
A graph $G$ is Helly if and only if the eccentricity function $e_G^M(\cdot)$ is unimodal on $G$ for every $M \subseteq V(G)$. 
\end{lemma}

%It was known earlier~\cite{FDraganPhD} that $2rad(G) - 1 \le diam(G) \le 2rad(G)$ holds for every Helly graph~$G$.
The following two results were earlier proven in~\cite{FDraganPhD}  only for $M=V(G)$ and then later extended in~\cite{newDDG2020} to all  $M \subseteq V(G)$. 

\begin{lemma}{\bf\cite{newDDG2020}}\label{lem:mDiamRadInHelly}
Let $G$ be a Helly graph. For any $M \subseteq V(G)$,
$2rad_G(M) - 1 \le diam_G(M) \le 2rad_G(M)$. Moreover, $rad_G(M) = \lceil diam_G(M) / 2 \rceil$.
\end{lemma}
%%%%%%%%%%%%%%%%%%%%%%%%%%%%%
%\begin{proof}\todo{\small we will remove proof when that  paper with Guillaume will be available}
%Consider in $G$ pairwise intersecting disks $D(x, \lceil diam_G(M) / 2 \rceil)$ for all $x \in M$.
%By the Helly property, there is a vertex $c \in V(G)$ such that $d_G(c,x) \leq \lceil diam_G(M) / 2 \rceil$ for all $x \in M$.
%Thus, $rad_G(M) \leq \ecc_G^M(c) \leq \lceil diam_G(M) / 2 \rceil \leq (diam_G(M)+1)/2$, proving the claim.
%\end{proof}

\begin{lemma}{\bf\cite{newDDG2020}}\label{lem:formula}
Let $G$ be a Helly graph. For every $v\in V(G)$ and $M \subseteq V(G)$, $\ecc_G^M(v)=d_G(v,C_G(M))+rad_G(M)$ holds.
\end{lemma} 
%%%%%%%%%%%%%%%%%%%%%%%%%%
%\begin{proof}\todo{\small same here}
%Let $M$ be any subset of $V(G)$. We will prove the formula by induction on $k=e_G^M(v)-rad_G(M)$. If $k=0$ then $e_G^M(v)=rad_G(M)$, i.e., $v\in C_G(M)$,  and the formula is trivially correct. Consider now a vertex $v$ with $e_G^M(v)>rad_G(M)$. By the triangle inequality, $e_G^M(v)\le d_G(v,C_G(M))+rad_G(M)$ always holds. As the eccentricity function $e_G^M(\cdot)$ is unimodal (see Lemma~\ref{lem:unimodalHelly}), there is a neighbor $u$ of $v$ with $e_G^M(v)>e_G^M(u)$. By induction hypothesis, $e_G^M(u)=d_G(u,C_G(M))+rad_G(M)$. Hence, by the triangle inequality, $e_G^M(v)\ge e_G^M(u)+1=d_G(u,C_G(M))+rad_G(M)+1\ge d_G(v,C_G(M))+rad_G(M)$. Combining two inequalities, we get $e_G^M(v)=d_G(v,C_G(M))+rad_G(M)$.
%\end{proof}

\begin{corollary}\label{cor:helly-center-to-disk}
For every Helly graph $G$, any subset $M\subseteq V(G)$ and any integers $\ell\ge 0$ and $k\ge 0$, 
$C^{\ell+k}_G(M)=D_G(C^{k}_G(M)+\ell)=D_G(C_G(M)+k+\ell).$ Furthermore, $diam_G(C^{\ell+k}_G(M))\le diam_G(C^{k}_G(M))+2\ell.$
\end{corollary}
\begin{proof} Consider a vertex $v$ with $\ecc_G^M(v)=k+\ell+rad_G(M)$ and a vertex $c$ in $C_G(M)$ closest to $v$. By Lemma~\ref{lem:formula}, $\ecc_G^M(v)=d_G(v,C_G(M))+rad_G(M)=d_G(v,c)+rad_G(M)=
k+\ell+rad_G(M)$. Hence, for every vertex $v$, $d_G(v,C_G(M))=k+\ell$ if and only if $\ecc_G^M(v)=k+\ell+rad_G(M)$.  

Let $u$ be a vertex on a shortest path from $v$ to $c$ at distance $k$ from $c$. By Lemma~\ref{lem:formula}, $\ecc_G^M(u)=d_G(u,c)+rad_G(M)=k+rad_G(M)$ and hence  $\ecc_G^M(v)=k+\ell+rad_G(M)=\ecc_G^M(u)+\ell$.  
Therefore, $\ecc_G^M(v)=k+\ell+rad_G(M)$ if and only if   $d_G(v,C^k_G(M))=d_G(v,u)=\ell$.

Let $x,y$ be vertices of $C^{\ell+k}_G(M)$ with $d_G(x,y)=diam_G(C^{\ell+k}_G(M))$. Since $d_G(v,C^k_G(M))\le \ell$ for each $v\in \{x,y\}$, by the triangle inequality, we have $d_G(x,y)\le d_G(x,C^k_G(M))+diam_G(C^k_G(M))+d_G(y,C^k_G(M))\le diam_G(C^k_G(M))+2\ell.$  
\end{proof}

%%%%%%%%%%%%%%%%%%%%%%%%%%%%%%%%%%%%%%%%%%%%%%%%%%%%%%%%%%%%
%%%%%%%%%%%%%%%%%%%%%%%%%%%%%%%%%%%%%%%%%%%%%%%%%%%%%%%%%%%%
% RELATION TO INJECTIVE HULLS
%%%%%%%%%%%%%%%%%%%%%%%%%%%%%%%%%%%%%%%%%%%%%%%%%%%%%%%%%%%%
%%%%%%%%%%%%%%%%%%%%%%%%%%%%%%%%%%%%%%%%%%%%%%%%%%%%%%%%%%%%
\section{Distances in injective hulls characterize weakly-Helly graphs}\label{inj-hull}
Recall that the \emph{injective hull} of $G$, denoted by $\calH(G)$, is a minimal Helly graph which contains $G$ as an isometric subgraph~\cite{Isbell1964/65,DRESS1984321}.
It turns out that~$\calH(G)$ is unique for every~$G$~\cite{Isbell1964/65}.
When $G$ is known by context, we often let $H := \calH(G)$.
% tight-span definition, with vector/function definition of vertices in injective hull
By an equivalent definition of an injective hull~\cite{DRESS1984321} (also called a tight span),
each vertex $f \in V(\calH(G))$ can be represented as a vector with nonnegative integer values $f(x)$ for each $x \in V(G)$,
such that the following two properties hold:
\begin{equation} \label{eq:atLeastDistance}
\forall x,y \in V(G) \ f(x) + f(y) \geq d_G(x,y) 
\end{equation}
\vspace*{-3mm}
\begin{equation} \label{eq:extremalFunctions}
\forall x \in V(G) \ \exists y \in V(G) \ f(x) + f(y) = d_G(x,y) 
\end{equation}

Additionally, there is an edge between two vertices $f,g \in V(\calH(G))$ if and only if their Chebyshev distance is 1, i.e., $\max_{x \in V(G)} \lvert f(x) -g(x) \rvert = 1$.
Thus, $d_H(f,g) = max_{x \in V(G)} \lvert f(x) - g(x) \rvert$.
Notice that if $f \in V(\calH(G))$, then $\{D(x, f(x)) : x \in V(G)\}$ is a family of pairwise intersecting disks.
For a vertex~$z \in V(G)$, define the distance function $d_z$ by setting $d_z(x)=d_G(z,x)$ for any $x \in V(G)$.
By the triangle inequality, each $d_z$ belongs to $V(\calH(G))$.
An isometric embedding of~$G$ into~$\calH(G)$ is obtained by mapping each vertex~$z$ of~$G$ to its distance vector $d_z$.

We classify every vertex $v$ in $V(\calH(G))$ as either a real vertex or a Helly vertex.
A vertex~$f \in V(\calH(G))$ is a \emph{real vertex} provided $f=d_z$ for some $z \in V(G)$, i.e., there is a one-to-one correspondence between $z \in V(G)$ and its representative real vertex $f \in V(\calH(G))$ which uniquely satisfies $f(z) = 0$ and $f(x) = d_G(z,x)$ for all $x \in V(G)$.
When working with~$\calH(G)$, we will use interchangeably the notation $V(G)$ to represent the vertex set in~$G$ as well as the vertex subset of~$\calH(G)$ which uniquely corresponds to the vertex set of~$G$.
Then, a vertex~$v \in V(\calH(G))$ is a real vertex if it belongs to $V(G)$ and a \emph{Helly vertex} otherwise.
Equivalently, a vertex~$h \in V(\calH(G))$ is a Helly vertex provided that $h(x) \ge 1$ for all $x \in V(G)$, that is, a Helly vertex exists only in the injective hull~$\calH(G)$ and not in~$G$.
We will show in Theorem~\ref{thm:closeRealVertex} that a graph $G$ is $\whp$-weakly-Helly if and only if the distance from any Helly vertex in $\calH(G)$ to a closest real vertex in $V(G)$ is no more than $\whp$. We recently learned that this result was independently discovered by Chalopin et al.~\cite{chalopin2020helly}. They use name {\em coarse Helly property} for  $\whp$-weakly-Helly property used here.   

The following properties will be useful to the main result of this section and to later sections. 
%%%%%%%%%%%%%%%%%%%%%%%%%%%%%%%%%%%%%%%%%%%%%%%%%%%%%%%%%%%%
% TODO: Properties of injective-hull following from graph theoretic terms
%%%%%%%%%%%%%%%%%%%%%%%%%%%%%%%%%%%%%%%%%%%%%%%%%%%%%%%%%%%%
A vertex $x$ is called a \emph{peripheral} vertex if
$I(y,x) \not\subset I(y,z)$ for some vertex $y$ and all vertices $z \neq x$.
We show next that the peripheral vertices of $\calH(G)$ are real vertices.
This adheres to the intuitive notion that an injective hull contains all of the Helly vertices ``between'' the vertices of 
$G$, so that the outermost vertices of~$\calH(G)$ are real.

\begin{proposition} \label{prop:outsidePointsAreReal}
  Peripheral vertices of $\calH(G)$ are real.
\end{proposition}
\begin{proof}
  By contradiction, suppose there is a peripheral Helly vertex $u \in V(\calH(G)) \setminus V(G)$.
  By definition, there is a vertex~$s \in V(\calH(G))$ such that, for all $x \in V(\calH(G))$, $I(s,u) \not\subset I(s,x)$.
  Let $k := d(u,s)$.
  Consider pairwise intersecting disks $D(s, k-1)$ and $D(x,1)$ for each $x \in D(u,1)$.
  By the Helly property, there exists vertex $w \in V(\calH(G))$ with $d(w,s)=k-1$ and $D(w,1) \supseteq D(u,1)$.
  By Lemma~\ref{lem:removeDominatingVertex}, $\calH(G) - \{u\}$ is Helly and is an isometric subgraph of $\calH(G)$.
  Since~$u$ is a Helly vertex, $G$ is an isometric subgraph of $\calH(G) - \{u\}$, a contradiction with the minimality of $\calH(G)$.
\end{proof}

Moreover, we show that any shortest path of $\calH(G)$ is a subpath of a shortest path between real vertices, which will later prove a useful property of injective hulls.

\begin{proposition} \label{prop:shortestPathSubsetOfReal}
Let $H$ be the injective hull of $G$.
For any shortest path $P(x,y)$, where $x,y \in V(H)$, there is a shortest path $P(x',y')$, where $x',y' \in V(G)$ such that $P(x',y') \supseteq P(x,y)$.
\end{proposition}
\begin{proof}
If $x$ and $y$ are both real vertices, then the proposition is trivially true.
Without loss of generality, let $y$ be a Helly vertex.
Consider a breadth-first search layering where $y$ belongs to layer $L_i$ of BFS($H, x$).
Let $y' \in L_k$ be a vertex with $y \in I(x,y')$ that maximizes $k=d_H(x,y')$.
Then, for any vertex $z \in V(\calH(G))$, $I(x,y') \not\subset I(x,z)$.
By Proposition~\ref{prop:outsidePointsAreReal}, $y' \in V(G)$.
If $x \notin V(G)$, then applying the previous step using BFS($H, y'$) yields vertex $x' \in V(G)$.
\end{proof}

We are now ready to prove the main result of this section.

\begin{theorem} \label{thm:closeRealVertex}
For any vertex $h \in V(\calH(G))$ there is a real vertex $v \in V(G)$ such that $d_{\calH(G)}(h,v) \leq \whp$ if and only if $G$ is an $\whp$-weakly-Helly graph.
\end{theorem}
\begin{proof}
% helly points in H(G) are within k from real vertex of G --> G is k-weakHelly
Suppose any Helly vertex in $H := \calH(G)$ is within distance at most $\whp$ from a vertex of $G$.
Consider in~$G$ a family of pairwise intersecting disks $\calF_G = \{ D_G(v, r(v)) : v \in S \subseteq V(G)\}$.
As~$H$ contains $G$ as an isometric subgraph, the disks $\calF_H = \{ D_H(v, r(v) : v \in S \}$ are also pairwise intersecting in~$H$.
By the Helly property, in~$H$ there is a vertex $x \in \bigcap_{v \in S} D_H(v, r(v))$.
By assumption, there is a vertex $u \in V(G)$ such that $d_H(u,x) \leq \whp$.
Thus $u \in \bigcap_{v \in S} D_G(v, r(v)+\whp)$ and so $G$ is $\whp$-weakly-Helly.

% G is k-weakHelly --> helly points in H(G) are within k from real vertex of G$
Assume now that $G$ is $\whp$-weakly-Helly.
Let $h \in V(H)$ be an arbitrary vertex represented as a vector with nonnegative integer values $h(x)$ for each $x \in V(G)$ satisfying conditions (\ref{eq:atLeastDistance}) and (\ref{eq:extremalFunctions}) from the definition of an injective hull.
Then, $\{D_G(x, h(x)) : x \in V(G)\}$ is a family of pairwise intersecting disks in~$G$.
By the $\whp$-weakly-Helly property, there is a real vertex $z \in V(G)$ belonging to $\bigcap_{x \in V(G)} D_G(x, h(x) + \whp)$.
To establish that $d_H(z,h) \leq \whp$ and complete the proof, we will show that $\max_{t \in V(G)} \lvert z(t) - h(t) \rvert \leq \whp$.

First, we show that $z(t) - h(t) \leq \whp$ for all $t \in V(G)$.
As~$z$ is a real vertex, by definition, $z(t) = d_G(z,t)$ and, by the $\whp$-weakly-Helly property  (recall that $z\in \bigcap_{x \in V(G)} D_G(x, h(x) + \whp)$), $z(t) \le h(t) + \whp$.
Next, we show that $h(t) - z(t) \leq \whp$ for all $t \in V(G)$.
Suppose that $h(x) - z(x) > \whp$ for some $x \in V(G)$.
Then, for all vertices $y \neq x$,
$h(x) + h(y) > z(x) + \whp + h(y) \geq z(x) + z(y) \geq d(x,y)$, a contradiction with condition (\ref{eq:extremalFunctions}).
\end{proof}

The injective hull is also useful to prove that $\whp$-weakly-Helly graphs are closed under pendant vertex addition. 
\begin{lemma}
Let $G+\{x\}$ be the graph obtained from $G$ by adding a vertex~$x$ pendant to any fixed vertex $v \in V(G)$.
Then, $\calH(G + \{x\}) = \calH(G) + \{x\}$.
\end{lemma}
\begin{proof}
As~$x$ is pendant to~$v$, then all $u \in V(G)$ have $d_{G+\{x\}}(u,x) = d_G(u,v) + 1$.
Let $H_1 := \calH(G + \{x\})$ and $H_2 := \calH(G) + \{x\}$. Note that $H_2$ is a Helly graph containing $G+\{x\}$ as an isometric subgraph. 
We first show that any $h \in V(H_1)$ also belongs to $V(H_2)$.
The statement clearly holds if $h$ is a real vertex, so assume that $h$ is a Helly vertex of $H_1$ represented as a vector with nonnegative integer values for each $u \in V(G+\{x\})$
satisfying conditions (\ref{eq:atLeastDistance}) and (\ref{eq:extremalFunctions}) from the definition of an injective hull.
We will show $h$ also satisfies the conditions under $G$.
By condition (\ref{eq:atLeastDistance}) under $G+\{x\}$, and since $G$ is isometric in $G+\{x\}$,
for all $u,y \in V(G)$, $h(u) + h(y) \ge d_{G+\{x\}}(u,y) = d_G(u,y)$. Thus, $h$ satisfies condition  (\ref{eq:atLeastDistance}) in $G$.
By condition (\ref{eq:extremalFunctions}) under $G+\{x\}$,
for every $u \in V(G)$ there is a vertex $y \in V(G +\{x\})$ with $h(u) + h(y) = d_{G+\{x\}}(u,y)$.
We claim that if $y=x$, then $h(x)=h(v)+1$ and so vertex~$v$ also satisfies $h(u) + h(v) = h(u) + h(y) - 1 = d_{G+\{x\}}(u,y) - 1 = d_G(u,v)$.
On one hand, $h(x) \le h(v) + 1$ since $h(u) + h(x) = d_{G+\{x\}}(u,x) = d_G(u,v) + 1 \le h(u) + h(v) + 1$.
On the other hand, let $z \in V(G) \cup \{x\}$ be a vertex such that $h(z) + h(v) = d_{G+\{x\}}(z,v)$. Note that $z\neq x$ as $h$ is not a real vertex and therefore $h(v) = d_{H_1}(h,v) >0$ and $h(z) = d_{H_1}(h,z) >0$.
Then, $h(x) \ge d_{G+\{x\}}(z,x) - h(z) = d_{G+\{x\}}(z,v) + 1 - h(z) = h(v) + 1$.
With the claim established, $h$ satisfies condition (\ref{eq:extremalFunctions}) in $G$. Thus, $h \in V(H_2)$ and $V(H_1) \subseteq V(H_2)$.
By minimality of $\calH(G)$, $V(H_1)=V(H_2)$.
%
%Now let $h$ be a Helly vertex of $V(H_2)$ and assign the value $h(x) := h(v) + 1$. We claim $h \in V(H_1)$.
%By condition (\ref{eq:atLeastDistance}) under $G$, and since $G$ is isometric in $G+\{x\}$, for all $u,y \in V(G)$, $h(u) + h(y) \ge d_G(u,y)$. As $h(x) \ge 1$, clearly $h(x) + h(x) \ge d_G(x,x)$.
%Since $G$ is isometric in $G+\{x\}$, any $u \in V(G)$ satisfies $h(u) + h(x) = h(u) + h(v) + 1 \ge d_G(u,v) + 1 = d_{G+\{x\}}(u,x)$. Thus, $h$ sastisfies condition (\ref{eq:atLeastDistance}) under $G+\{x\}$.
%By condition (\ref{eq:extremalFunctions}) under $G$, for all $u \in V(G)$ there is a $y \in V(G)$ with $h(u) + h(y) = d_G(u,y)$.
%As vertex~$v$ has a vertex~$y \in V(G)$ with $h(v) + h(y) = d_G(v,y)$, then $h(x) + h(y) = h(v) + h(y) + 1 = d_G(v,y) + 1 = d_{G+\{x\}}(x,y)$.
%Thus, $h$ satisfies condition (\ref{eq:extremalFunctions}) under $G+\{x\}$, so $h \in V(H_1)$.
\end{proof}

\begin{corollary}\label{cor:weakHellyPreservedByPendant}
Let $G+\{x\}$ be the graph obtained from $G$ by adding a pendant vertex~$x$ adjacent to any fixed vertex $v \in V(G)$. Then, $\whp(G) = \whp(G+\{x\})$.
\end{corollary}

%%%%%%%%%%%%%%%%%%%%%%%%%%%%%%%%%%%%%%%%%%%%%%%%%%%%%%%%%%%%
%%%%%%%%%%%%%%%%%%%%%%%%%%%%%%%%%%%%%%%%%%%%%%%%%%%%%%%%%%%%
% DIAM, RAD, ECCENTRICITIES
%%%%%%%%%%%%%%%%%%%%%%%%%%%%%%%%%%%%%%%%%%%%%%%%%%%%%%%%%%%%
%%%%%%%%%%%%%%%%%%%%%%%%%%%%%%%%%%%%%%%%%%%%%%%%%%%%%%%%%%%%
\section{Diameter, radius, and all eccentricities}\label{sec:diam-rad-ecc}
% define diameter
% define center
% define eccentricity
%Recall the eccentricity of a vertex~$\ecc_G(v) = \max_{u \in V(G)} d_G(v,u)$ and, if the graph~$G$ is clear by context, the subindex is omitted.
We establish several bounds on all eccentricities, and the diameter and radius in particular, for $\whp$-weakly-Helly graphs.
As an intermediate step, we relate the diameter, radius, and all eccentricities in~$G$ to their counterparts in~$\calH(G)$. 

%%%%%%%%%%%%%%%%%%%%%%%%%%%%%%%%%%%%%%%%%%%%%%%%%%%%%%%%%%%%
% DIAM, RAD, ECC w.r.t H(G)
%%%%%%%%%%%%%%%%%%%%%%%%%%%%%%%%%%%%%%%%%%%%%%%%%%%%%%%%%%%%
\subsection{Eccentricities and centers}
%Relation between vertex eccentricities in~$G$ and in~$\calH(G)$}
First, we provide a few immediate consequences of Proposition~\ref{prop:outsidePointsAreReal} which establishes that farthest vertices in~$\calH(G)$ are real vertices. That is, for any $v \in V(\calH(G))$, $F_{\calH(G)}(v) \subseteq V(G)$. It follows that all eccentricities in~$G$ and the diameter of~$G$ are preserved in $\calH(G)$, including with respect to any subset $M \subseteq V(G)$ of real vertices. 

% definition in G vs H : considers vertices of G or H
% definition w.r.t. M in G vs H : considers vertices of G or H, but M includes only vertices of G.
\begin{proposition} \label{prop:eccentricityEqual}
Let $H$ be the injective hull of $G$.
For any $M \subseteq V(G)$ and $v \in V(G)$, $\ecc_G^M(v) = \ecc_H^M(v)$.
Moreover, $\ecc_G(v) = \ecc_H(v)$.
\end{proposition}

\begin{proposition} \label{prop:diameterEqual}
Let $H$ be the injective hull of $G$.
For any $M \subseteq V(G)$, $diam_G(M) = diam_H(M)$.
Moreover, $diam(G) = diam(H)$.
\end{proposition}

%%%%%%%%%%%%%%%%%%%%  Heather added
%For some set of vertices $M \subseteq V(G)$ and vertex $v \in V(G)$ we will use
%%%the notations $\ecc_G(v)$ and $\ecc_H(v)$ interchangeably. Moreover, for a fixed $M$, we use
%the shorthand notation $\ecc^M(v)$ to represent the equivalent values $\ecc_G^M(v)$ and $\ecc_H^M(v)$.
%For a vertex set $M \subseteq V(G)$, we use $diam(M)$ to denote the equivalent values $diam_G(M)$ or $diam_H(M)$.
%%%; similarly, we use $diam(G)$ and $diam(H)$ interchangeably.
%%%%%%%%%%%%%%%%%%%%

\begin{proposition} \label{prop:radiusInH}
Let $H$ be the injective hull of an $\whp$-weakly-Helly graph~$G$.
For any $M \subseteq V(G)$,
$rad_G(M) - \whp \le rad_H(M) \le rad_G(M)$.
%$rad_H(M) \le rad_G(M)$ and  $rad_G(M) \leq rad_H(M) + \whp$.
In particular, $rad(G) - \whp \le rad(H) \le rad(G)$.
\end{proposition}
\begin{proof}
By Proposition~\ref{prop:eccentricityEqual}, the eccentricity of any vertex of~$G$ is preserved in~$H$.
Hence, $rad_H(M) \le rad_G(M)$.
Consider %in $H$ a family of pairwise intersecting disks~$\{D_H(x, rad_H(M)) : x \in M\}$.
%By the Helly property, there is a vertex $h \in V(H)$ such that $d_H(h,x) \leq rad_H(M)$ for all $x \in M$.
any vertex $h \in C_H(M)$. We have $d_H(h,x) \leq rad_H(M)$ for all $x \in M$.
By Theorem~\ref{thm:closeRealVertex}, there is a real vertex $v \in V(G)$ such that $d_H(v,h) \leq \whp$.
By the triangle inequality, $rad_G(M) \leq \ecc_G^M(v) \leq max_{x \in M}\{d_H(v,h) + d_H(h,x)\} \le rad_H(M) + \alpha$.
\end{proof}

Though the diameter is preserved, the radius in~$\calH(G)$ may be smaller than the radius in~$G$ by at most $\whp(G)$. In this case, the radius in~$\calH(G)$ is realized by Helly vertices which are not present in~$G$, resulting in different centers in~$\calH(G)$ and~$G$. In what follows, we establish that, for any $M \subseteq V(G)$, any vertex of $C_G(M)$ is close to a vertex of $C_H(M)$. Moreover, the diameter of the center of $M$ in~$G$ is at most the diameter of the center of $M$  in the injective hull $\calH(G)$ plus $2\whp(G)$.

\begin{lemma} \label{lem:ClGisSubsetOfDiskAroundCH}
Let~$H$ be the injective hull of an $\whp$-weakly-Helly graph~$G$.
For any $M \subseteq V(G)$ and integer $\ell \ge 0$, $C_G^{\ell}(M) \subseteq D_H(C_H(M), \whp+\ell)=D_H(C^{\ell}_H(M),\alpha)$.
\end{lemma}
\begin{proof}
Consider any vertex $x \in C_G^{\ell}(M)$.
By Proposition~\ref{prop:radiusInH}, $\ecc_G^M(x) \leq rad_G(M) + \ell \leq rad_H(M) + \whp + \ell$.
By Proposition~\ref{prop:eccentricityEqual} and Lemma~\ref{lem:formula} (since $H$ is Helly), $\ecc_G^M(x) = \ecc_H^M(x) = rad_H(M) + d_H(x, C_H(M))$. 
Therefore, $d_H(x,C_H(M)) \leq \whp + \ell$. % and $x \in D_H(C_H(M), \whp+\ell)$.
By Corollary~\ref{cor:helly-center-to-disk}, 
$D_H(C_H(M), \whp+\ell)=D_H(C^{\ell}_H(M),\alpha).$
\end{proof}

\begin{theorem}\label{thm:diameterOfCenters}
Let $H$ be the injective hull of an $\whp$-weakly-Helly graph~$G$.
For any integer $\ell \ge 0$ and any $M \subseteq V(G)$,
$diam_G(C_G^{\ell}(M)) - 2\whp  \le diam_H(C_H^{\ell}(M)) \le diam_G(C_G^{\whp + \ell}(M)) + 2\whp$.
%In particular, $diam_G(C_G(M)) -  2\whp \le diam_H(C_H(M)) \le diam_G(C_G^{\whp}(M)) + 2\whp$.
\end{theorem}
\begin{proof}
% for any M \subseteq V(G)
Let $d_H(u,v) = diam_H(C_H^\ell(M))$ for some $u,v \in C_H^\ell(M)$. % V(H)$.
By Theorem~\ref{thm:closeRealVertex},
there is a real vertex $u^* \in V(G)$ at distance at most $\whp$ from $u$.
By Proposition~\ref{prop:eccentricityEqual} and Proposition~\ref{prop:radiusInH},
$\ecc_G^M(u^*) = \ecc_H^M(u^*) \le \ecc_H^M(u) + \whp \le rad_H(M) + \ell + \whp \le rad_G(M) + \ell + \whp$.
Similarly, there is a real vertex $v^* \in V(G)$ at distance at most $\whp$ from $v$
with $\ecc_G^M(v^*) \le rad_G(M) + \ell + \whp$.
Both vertices $v^*,u^*$ belong to $C_G^{\whp + \ell}(M)$.
By the triangle inequality,
$diam_H(C_H^\ell(M)) = d_H(u,v) \le 2\whp + d_H(u^*,v^*) \le 2\whp + diam_G(C_G^{\whp + \ell}(M))$.

On the other hand, by Lemma~\ref{lem:ClGisSubsetOfDiskAroundCH},
any vertex of $C_G^{\ell}(M)$ has distance at most $\whp$ to $C^{\ell}_H(M)$.
Let $x,y \in C_G^{\ell}(M)$ %V(G)$ 
such that $d_G(x,y) = diam_G(C_G^{\ell}(M))$.
By the triangle inequality, $d_G(x,y)=d_H(x,y) \le diam_H(C^{\ell}_H(M)) + 2\whp$. A small graph depicted on Figure \ref{fig:example_sharpness} demonstrates that this inequality is tight. 
\end{proof}

\begin{center}
\includegraphics[scale=1]{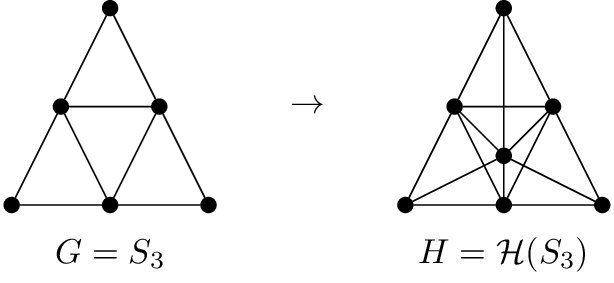}
\captionof{figure}{A graph $G$ (left) and its injective hull $H$ (right) which show that inequalities in Proposition~\ref{prop:radiusInH} and Theorem~\ref{thm:diameterOfCenters} are tight: $\whp(G)=1$, $rad(G)=2$, $rad(H)=1$, $diam(C(H))=0$, $diam(C(G))=2$.}\label{fig:example_sharpness}
\end{center}

\begin{corollary}\label{cor:centersInclusions}
For each $\whp$-weakly-Helly graph $G$ with the injective hull $H={\calH(G)}$ and every $\ell \ge 0$,
$C^{\ell}(H) \cap V(G) \subseteq C^{\ell}(G) \subseteq C^{\ell+\whp}(H) \cap V(G).$
\end{corollary}

\begin{center}
\includegraphics[scale=0.8]{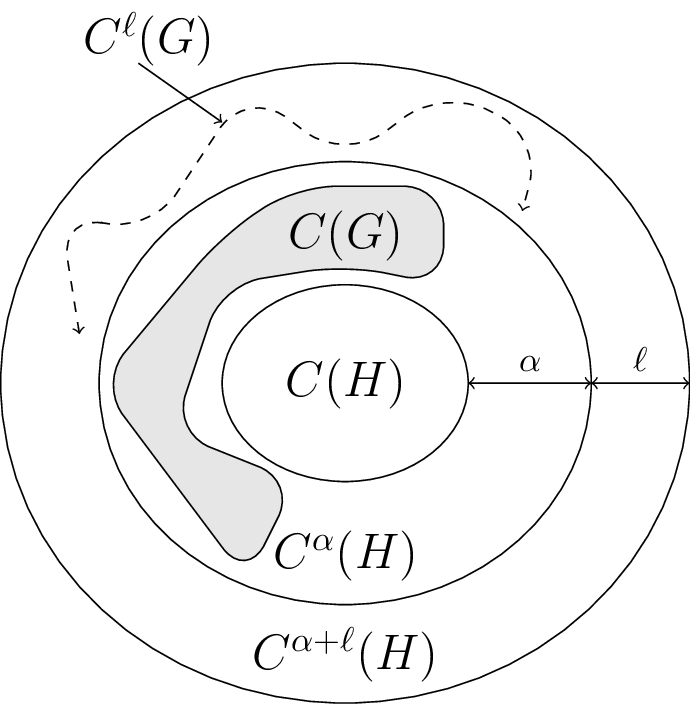}
\captionof{figure}{Inclusions for sets $C(G)$, $C(H)$, $C^\whp(H)$, $C^\ell(G)$, and $C^{\whp+\ell}(H)$ in $H= \calH(G)$.}\label{fig:example_inclusionsCenters}
\end{center}

\begin{proof}
Let $v \in C^{\ell}(H) \cap V(G)$.
By Proposition~\ref{prop:radiusInH} and Proposition~\ref{prop:eccentricityEqual}, $e_G(v) = e_H(v) \le rad(H) + \ell \le rad(G) + \ell$. Hence, $v \in C^{\ell}(G)$.
Consider now a vertex $v \in C^{\ell}(G)$.
It follows from Lemma~\ref{lem:ClGisSubsetOfDiskAroundCH}
that $d_H(v,C(H)) \le \ell + \whp$.
By Lemma~\ref{lem:formula}, $e_H(v) = rad(H) + d_H(v,C(H)) \le rad(H) + \ell + \whp$.
Thus, $v \in C^{\ell+\whp}(H)$.
\end{proof}

The results from this section on inclusions for central sets are illustrated in Figure~\ref{fig:example_inclusionsCenters}.

%%%%%%%%%%%%%%%%%%%%%%%%%%%%%%%%%%%%%%%%%%%%%%%%%%%%%%%%%%%%
% BOUNDS ON WEAK-HELLYNESS
%%%%%%%%%%%%%%%%%%%%%%%%%%%%%%%%%%%%%%%%%%%%%%%%%%%%%%%%%%%%
\subsection{Relation between \hg{} and diameter, radius and graph powers}
We obtain a few lower and upper bounds on the \hg{} $\whp(G)$.
The first follows directly from Proposition~\ref{prop:radiusInH}.

% this is less interesting.
\begin{corollary}\label{cor:lowerBoundOnWHByRad}
Let $H := \calH(G)$. For any $M \subseteq V(G)$, $\whp(G) \geq rad_G(M) - rad_H(M)$. %, where $\whp$ is the \hg{} of~$G$.
\end{corollary}

If $\calH(G)$ were given, one could compute~$\whp(G)$ by computing the maximum distance from a Helly vertex to a closest real vertex. %Unfortunately, computing $\calH(G)$ can be computationally expensive even for small graphs.
However, we provide in Corollary~\ref{cor:whpBoundedByDiam} and Lemma~\ref{lem:whpLowerBoundByDiam2Rad} upper and lower bounds which do not necessitate computing the injective hull.

\begin{corollary} \label{cor:whpBoundedByDiam}
Any graph~$G$ is~$\whp$-weakly-Helly for $\whp \le \lfloor diam(G) / 2 \rfloor$.
\end{corollary}
\begin{proof}
%% do this without Th 1 or Prop 4.
By contradiction, suppose $\whp(G) > \lfloor diam(G) / 2 \rfloor$.
Let $H := \calH(G)$.
By Theorem~\ref{thm:closeRealVertex}, in $H$ there exists a Helly vertex~$u$ with $d_H(u,v) > \lfloor diam(G) / 2 \rfloor$ for all $v \in V(G)$. 
By Proposition~\ref{prop:shortestPathSubsetOfReal}, $u$ belongs to a shortest~$(x,y)$-path for two real vertices~$x,y \in V(G)$.
As $G$ is isometric in $H$,
$d_G(x,y) = d_H(x,y) = d_H(x,u) + d_H(u,y) \ge  2(\lfloor diam(G) / 2 \rfloor + 1) > diam(G)$, a contradiction with $d_G(x,y) \le diam(G)$.
\end{proof}

\begin{lemma}\label{lem:whpLowerBoundByDiam2Rad}
Let $H$ be the injective hull of~$G$, $M$ be any subset of $V(G)$ and $k\ge 0$ be an integer.
If $diam_G(M) = 2rad_G(M) - k$, then $rad_G(M) = rad_H(M) + \lfloor k/2 \rfloor$.
Moreover, $\whp(G) \ge \lfloor (2rad_G(M) - diam_G(M)) / 2 \rfloor$, that is, $2rad_G(M) \ge diam_G(M) \geq 2rad_G(M) - 2\whp(G) -1$.
\end{lemma}
\begin{proof}
By Proposition~\ref{prop:diameterEqual}, $2rad_G(M) - k = diam_G(M) = diam_H(M)$.
Since $H$ is Helly, by Lemma~\ref{lem:mDiamRadInHelly}, $rad_H(M) = \lceil (2rad_G(M) - k)/2 \rceil = rad_G(M) - \lfloor k/2 \rfloor$.
By Corollary~\ref{cor:lowerBoundOnWHByRad}, $\whp(G) \ge  \lfloor k/2 \rfloor$.
\end{proof}

\begin{remark} %\hcomment{\small remove once we characterize $\whp(G)$ wrt any M}
Observe that the \hg{} of a graph~$G$ can be much larger than $rad(G) - rad(H)$ and $\lfloor (2rad(G) - diam(G))/2 \rfloor$.
Consider a graph~$G$ formed by a cycle $C_{4k}$ of size $4k$ and two paths of length $k$ each connected to opposite ends of the cycle, as illustrated in Fig.~\ref{fig:example_alphaFarFromRadiusDifference}.
By Lemma~\ref{lem:whpLowerBoundByDiam2Rad}, $\whp(C_{4k}) \ge k$. 
By Corollary~\ref{cor:weakHellyPreservedByPendant}, $\whp(G)= \whp(C_{4k}) \ge k$.
However, $diam(G) = 4k$ and $rad(G) = 2k$, and by Lemma~\ref{lem:whpLowerBoundByDiam2Rad}, $rad(G) = rad(H)$.
In this case, $rad(G) - rad(H) = \lfloor (2rad(G) - diam(G))/2 \rfloor = 0$.

\begin{center}
\includegraphics[scale=0.9]{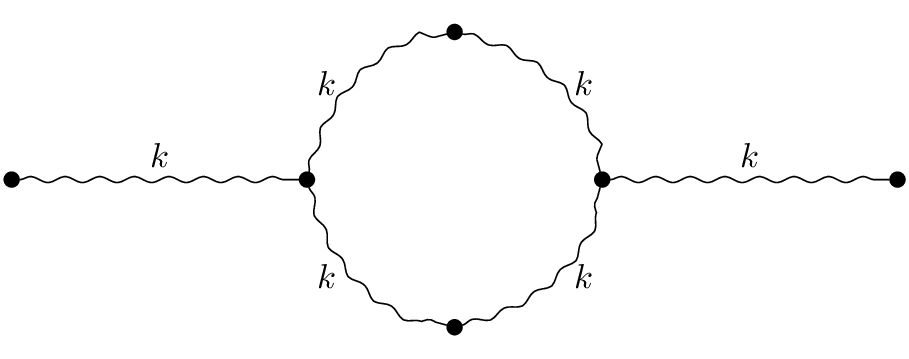}
\captionof{figure}{A graph~$G$ with \hg{} $k$ and $rad(G)=rad(H)=2k$.}\label{fig:example_alphaFarFromRadiusDifference}
\end{center}
\end{remark}

\begin{corollary}
Let $H$ be the injective hull of~$G$ and $M \subseteq V(G)$.
Then, $diam_G(M) \ge 2rad_G(M) - 1$ if and only if $C_G(M) \subseteq C_H(M)$. % equality here is not possible - consider C_4 with opposing ends having a pendant vertex. Set M=diam(G). |C_G(M)=2| but |C_H(M)=3|.
%\hcomment{\small D - remove? All M iff Helly}
\end{corollary}
\begin{proof}
If $diam_G(M) \ge 2rad_G(M) - 1$, then by Lemma~\ref{lem:whpLowerBoundByDiam2Rad},
$rad_G(M)=rad_H(M)$.
Thus, any $v \in C_G(M)$ has $e_H^M(v) \le rad_H(M)$.
On the other hand, if $C_G(M) \subseteq C_H(M)$, since eccentricities are preserved in $H$ by Proposition~\ref{prop:eccentricityEqual},
then also $rad_G(M) = rad_H(M)$.
Since~$H$ is Helly, by Lemma~\ref{lem:mDiamRadInHelly} and Proposition~\ref{prop:diameterEqual}, $diam_G(M) = diam_H(M) \ge 2rad_H(M) - 1 = 2rad_G(M) - 1$ holds.
\end{proof}

Interestingly, the \hg{} of a graph~$G$ decreases in powers of~$G$.
\begin{lemma}
Let~$G$ be an $\whp$-weakly-Helly graph.
For every integer $k \ge 1$, $G^{k}$ is $\lceil \whp / k \rceil$-weakly-Helly.
\end{lemma}
\begin{proof}
Let $\calF = \{D_{G^k}(v, r(v)) : v \in S\}$ be a system of disks  which pairwise intersect in $G^k$ so that any two vertices $u,v \in S$ satisfy $d_{G^k}(u,v) \le r(u) + r(v)$.
Then, their distances in $G$ satisfy $d_G(u,v) \le k(r(u) + r(v))$.
Consider now a corresponding system of disks in $G$ centered at same vertices, defined as $\calM = \{D_G(v, kr(v)) : v \in S\}$.
$\calM$ is a family of pairwise intersecting disks of $G$ as any two vertices $u,v \in S$
satisfy $d_G(u,v) \le kr(u) + kr(v)$.
As $G$ is $\whp$-weakly-Helly, there exists a common vertex $z \in \cap \{D_G(v, kr(v) + \whp) : v \in S\}$.
Since, for any $v \in S$, $d_G(z,v) \le kr(v) + \whp$,
necessarily, $d_{G^k}(z,v) \le r(v) + \lceil \whp / k \rceil$.
Hence, vertex~$z$ intersects all disks of~$\calF$ when the radii of each disk is extended by  $\lceil \whp / k \rceil$.
Therefore, $G^k$ is $\lceil \whp / k \rceil$-weakly-Helly.
\end{proof}

The results of this subsection are summarized in Theorem~\ref{thm:hGapBounds}.
\begin{theorem}\label{thm:hGapBounds}
Let~$G$ be an arbitrary graph. Then, the following holds:
\begin{itemize}[nolistsep,noitemsep]
    \item[i)] $\lfloor (2rad(G) - diam(G)) / 2 \rfloor \le \whp(G) \le \lfloor diam(G) / 2 \rfloor$, and
    \item[ii)] $\whp(G^k) \le \lceil \whp(G) / k \rceil$.
\end{itemize}
\end{theorem}

This Theorem~\ref{thm:hGapBounds} generalizes some known results for Helly graphs. Recall that, in Helly graphs, $\lfloor (2rad(G) - diam(G)) / 2 \rfloor =0$ holds and that every power of a Helly graph is a Helly graph as well~\cite{FDraganPhD}.  We know also (see also Section \ref{sec:cases}) that the Helly-gap of a chordal graph or a distance-hereditary graph is at most 1 and the Helly-gap of a $\delta$-hyperbolic graph is at most $2\delta$. Hence, Theorem~\ref{thm:hGapBounds} generalizes some known results on those graphs, too. For every chordal graph $G$ as well as for every distance-hereditary graph $G$, $\lfloor (2rad(G) - diam(G)) / 2 \rfloor \le 1$ holds~\cite{chepoi:center-triang,yushmanovMetricGraphProperties,10.1007/3-540-58218-5_34}.  For every $\delta$-hyperbolic graph  $G$, $\lfloor (2rad(G) - diam(G)) / 2 \rfloor \le 2\delta$ holds~\cite{Chepoi_2008,Dragan2018RevisitingRD}.  
%diam(G)≥2rad(G)−4δ−1 

%%%%%%%%%%%%%%%%%%%%%%%%%%%%%%%%%%%%%%%%%%%%%%%%%%%%%%%%%%%%
% UNIMODALITY
%%%%%%%%%%%%%%%%%%%%%%%%%%%%%%%%%%%%%%%%%%%%%%%%%%%%%%%%%%%%
\subsection{The eccentricity function is almost unimodal in $\whp$-weakly-Helly graphs}
The \emph{locality} $loc(v)$ of a vertex~$v$ is defined as the minimum distance from~$v$ to a vertex of strictly smaller eccentricity. Recall that the eccentricity function $\ecc_G(\cdot)$ is {\em unimodal} in $G$ if every vertex $v \in V(G) \setminus C(G)$ has a
neighbor~$u$ such that $\ecc_G(u) < \ecc_G(v)$, i.e., $loc(v) = 1$.
In a graph~$G$ with a unimodal eccentricity function, any local minimum of the eccentricity function (i.e., a vertex whose eccentricity is not larger than the eccentricity of any of its neighbors) is the global minimum of the eccentricity function in~$G$ (i.e., it is a central vertex).
Recall that Helly graphs are characterized by the property that every eccentricity function $\ecc_G^M$ is unimodal for any $M \subseteq V(G)$; therefore, $\ecc_G^M(v) = d(v,C_G(M)) + rad_G(M)$ holds (see Lemma~\ref{lem:formula}).
A natural question for $\whp$-weakly-Helly graphs is whether similar results on the unimodality of the eccentricity function hold up to a function of~$\whp$, that is, if any vertex $v \in V(G) \setminus C^{\whp}(G)$ has $loc(v) \le f(\whp)$.
The following lemmas answer in the positive.

\begin{lemma} \label{lem:eccentricityIsAlmostUnimodal}
Let $G$ be an $\whp$-weakly-Helly graph and let $M \subseteq V(G)$.
If there is a vertex~$v \in V(G)$ such that $\ecc_G^M(v) > rad_G(M) + \whp$, then there is a vertex $u \in D_G(v, 2\whp + 1)$ with $\ecc_G^M(u) < \ecc_G^M(v)$.
\end{lemma}
\begin{proof}
Let $S = M \cup \{v\}$.
Consider in~$G$ a system of disks $\calF = \{D(u, \rho_u) : u \in S \}$,
where the radii are defined as $\rho_w = \ecc_G^M(v) - 1 - \whp$ for any $w \in M$, and $\rho_v = \whp + 1$.
We assert that all disks of $\calF$ are pairwise intersecting.
Clearly for any $w \in M$, disks $D(v, \rho_v)$ and $D(w, \rho_w)$ intersect as 
$d(w,v) \leq \ecc_G^M(v)$.
		%= (\ecc_G^M(v) - 1 - \whp) + (\whp + 1)
		%= \rho_w + \rho_v$.
We now show that for any $w,w' \in M$ the disks $D(w, \rho_w)$ and $D(w', \rho_{w'})$ intersect.

Consider a vertex $c \in C_G(M)$.
By choice of $v$, $\ecc_G^M(v) \geq rad_G(M) + \whp + 1 = \ecc_G^M(c) + \whp +1$.
By the triangle inequality, for any two vertices $w,w' \in M$ we have
$d(w, w') \leq d(w,c) + d(w', c)
		  \leq \ecc_G^M(c) + \ecc_G^M(c)
		  \leq (\ecc_G^M(v) - 1 - \whp) + (\ecc_G^M(v) - 1 - \whp)
		  = \rho_w + \rho_{w'}$.
Then, by the $\whp$-weakly-Helly property, the system $\calF$ of pairwise intersecting disks
has a common intersection when radii of all disks are extended by $\whp$.
Therefore, there is a vertex $u$ such that $d(u,v) \leq 2\whp + 1$ and $d(u,w) \leq \ecc_G^M(v) - 1$ for all $w \in M$.
\end{proof}

For $\alpha =0$ we obtain a result known for Helly graphs.  As $\delta$-hyperbolic graphs are $(2\delta)$-weakly-Helly (this follows from a result in \cite{DBLP:conf/approx/ChepoiE07}, see also Section \ref{sec:cases}), Lemma \ref{lem:eccentricityIsAlmostUnimodal} also greatly generalizes a result known for $\delta$-hyperbolic graphs (see~\cite{ourManuscriptHyperbolicityTerrain}). 

\begin{lemma} \label{lem:ifTooFarFromCH}
Let $H$ be the injective hull of an $\whp$-weakly-Helly graph~$G$, and let $M \subseteq V(G)$.
If there is a vertex $v \in V(G)$ such that $d_H(v, C_H(M)) > \whp$, then there is a real vertex $u \in D_G(v,2\whp+1)$ such that $\ecc_G^M(u) < \ecc_G^M(v)$.
\end{lemma}
\begin{proof}
Let $c \in C_H(M)$ be a vertex closest in $H$ to $v$, and assume that $d_H(v, c) \geq \whp + 1$.
Let $w \in V(H)$ be a vertex on a shortest $(v,c)$-path of $H$  such that $d_H(v,w) = \whp + 1$.
Since $H$ is Helly, by Proposition~\ref{prop:eccentricityEqual} and Lemma~\ref{lem:formula}, $\ecc_G^M(v) = \ecc_H^M(v) = d_H(v, C_H(M)) + rad_H(M) = d_H(v,w) + d_H(w,c) + rad_H(M) = \whp + 1 + \ecc_H^M(w)$.
Therefore, $\ecc_H^M(w) = \ecc_G^M(v) - \whp - 1$.
Since $G$ is $\whp$-weakly-Helly, by Theorem~\ref{thm:closeRealVertex},  there is a vertex $u \in V(G)$ such that $d_H(u,w) \leq \whp$.
Hence, $\ecc_G^M(u) = \ecc_H^M(u) \leq \ecc_H^M(w) + \whp \leq \ecc_G^M(v) - 1$. See Fig.~\ref{fig:proofifTooFarFromCH} for an illustration.
\end{proof}

\begin{center}
\includegraphics[scale=1]{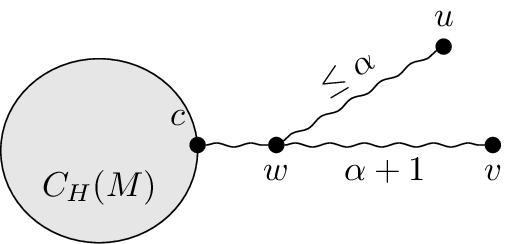}
\captionof{figure}{Illustration to the proof of Lemma~\ref{lem:ifTooFarFromCH}.}\label{fig:proofifTooFarFromCH}
\end{center}

Theorem~\ref{thm:unimodality} summarizes the results of Lemma~\ref{lem:eccentricityIsAlmostUnimodal} and Lemma~\ref{lem:ifTooFarFromCH}.
\begin{theorem}\label{thm:unimodality}
Let $H$ be the injective hull of an $\whp$-weakly-Helly graph~$G$, and let $M \subseteq V(G)$.
If there is a vertex $v \in V(G)$ such that $\ecc_G^M(v) > rad_G(M) + \whp$ or $d_H(v, C_H(M)) > \whp$, then there is a real vertex $u \in D_G(v,2\whp+1)$ such that $\ecc_G^M(u) < \ecc_G^M(v)$.
\end{theorem}

%%%\section{Approximation of all eccentricities}
\section{Estimating all eccentricities}\label{sec:estimatingEcc}
This section provides upper and lower bounds on the eccentricity of a vertex~$v$ based on a variety of conditions: the distance from~$v$ to a closest almost central vertex (i.e., a closest vertex with eccentricity at most $rad_G(M) + \whp$) and whether $v$ is a farthest vertex from some other vertex and if $diam(C_G^{2\whp}(M))$ is bounded. We also prove that any $\whp$-weakly-Helly graph has an eccentricity approximating spanning tree where the additive approximation error depends on the diameter of the set $C^{\whp}(G)$.
Finally, we describe the eccentricity terrain in $\whp$-weakly-Helly graphs, that is, how vertex eccentricities change along vertices of a shortest path to $C_G^{\whp}(M)$. 

\subsection{Using distances to $C_G^{\whp}(M)$}

\begin{lemma} \label{lem:eccUBound}
Let~$G$ be a graph, $M \subseteq V(G)$, and $k \ge 0$.
For every vertex $x \in V(G)$, $\ecc_G^M(x) \le d(x,C_G^k(M)) + rad_G(M) + k$ holds.
\end{lemma}
\begin{proof}
Let $\ecc_G^M(x) = d(x,y)$ where $y \in F_G^M(x)$,
and let $x_p \in C_G^k(M)$ be closest to~$x$.
By the triangle inequality,
$\ecc_G^M(x) \le d(x,x_p) + d(x_p,y) \le d(x,C_G^k(M)) + rad_G(M) + k$.
\end{proof}

\begin{lemma} \label{lem:eccLBound}
Let $G$ be an $\whp$-weakly-Helly graph and $M \subseteq V(G)$.
For every vertex $x \in V(G)$, $\ecc_G^M(x) \ge d(x,C_G^{\whp}(M)) + rad_G(M) - \whp$ holds.
\end{lemma}
\begin{proof}
Let $H := \calH(G)$ and $h \in C_H(M)$ be a closest vertex to~$x$ in~$H$.
By Lemma~\ref{lem:formula} and Proposition~\ref{prop:eccentricityEqual}, $d_H(x,h) = d_H(x,C_H(M)) = \ecc_H^M(x) - rad_H(M) = \ecc_G^M(x) - rad_H(M) \ge rad_G(M) - rad_H(M)$.
Then, let~$y$ be a vertex on a shortest $(x,h)$-path with $d_H(h,y) = rad_G(M) - rad_H(M)$.
By Lemma~\ref{lem:formula}, $\ecc_H^M(y) = d_H(y,C_H(M)) + rad_H(M) = d_H(y,h) + rad_H(M) = rad_G(M)$.
%By Proposition~\ref{prop:eccentricityEqual},
%$\ecc_G^M(x) = \ecc_H^M(x) = d_H(x,C_H(M)) + rad_H(M) = d_H(x,y) + rad_G(M)$.
By Theorem~\ref{thm:closeRealVertex}, there is a real vertex $y^* \in V(G)$ such that $d_H(y,y^*) \leq \whp$.
Applying Proposition~\ref{prop:eccentricityEqual} one obtains $\ecc_G^M(y^*) = \ecc_H^M(y^*) \le rad_G(M) + \whp$.
By the triangle inequality,
$d_H(x,y) \ge d_H(x,y^*) - d_H(y,y^*)  \geq d(x,C_G^\whp(M)) - \whp$.
Thus, $\ecc_G^M(x) = d_H(x,h) + rad_H(M) = d_H(x,y) + rad_G(M) \geq  d(x,C_G^\whp(M)) + rad_G(M) - \whp$.
\end{proof}

Applying Lemma~\ref{lem:eccUBound} with $k=\whp$ and Lemma~\ref{lem:eccLBound}, we obtain the following approximation
of eccentricities in $\whp$-weakly-Helly graphs.
\begin{theorem}\label{thm:eccentricities}
Let $G$ be an $\whp$-weakly-Helly graph. For every $M \subseteq V(G)$, every $x \in V(G)$ satisfies
$$d(x,C_G^\whp(M)) + rad_G(M) - \whp   \leq \ecc_G^M(x) \leq d(x,C_G^\whp(M)) + rad_G(M) + \whp.$$
\end{theorem}

If $C_G^\whp(M)$ is known in advance, then one obtains a linear time additive $2\whp$-approximation of all eccentricities by using a breadth-first search from $C_G^\whp(M)$ ($BFS(C_G^\whp(M))$), to obtain all distances to $C_G^\whp(M)$, and a $BFS(c)$ from any fixed $c \in C_G^\whp(M)$ that has a neighbor not in $C_G^\whp(M)$, to compute $rad_G(M) + \whp$. 
Then, for each vertex~$v$, set an approximate eccentricity $\hat{e}^M(v) = d_G(v,C_G^\whp(M)) + rad_G(M) + \whp$.
By Theorem~\ref{thm:eccentricities}, $e_G^M(v) \le \hat{e}^M(v) \le e_G^M(v) + 2\whp$.

%------------
By Corollary~\ref{cor:helly-center-to-disk}, any Helly graph~$G$ satisfies $C_G^\ell(M) = D(C_G(M), \ell)$. This also extends to $\whp$-weakly-Helly graphs in the following way. 
%applies to $\whp$-weakly-Helly graphs up to a function of $\whp$.

\begin{corollary} \label{lem:ClGisSubsetOfDiskAroundCkG}
Let~$G$ be an $\whp$-weakly-Helly graph.
For any $M \subseteq V(G)$ and any integer $\ell \ge 0$, $C_G^\ell(M) \subseteq D_G(C_G^\whp(M), \whp+\ell)$.
\end{corollary}
\begin{proof}
If $x \in C_G^\ell(M)$ then,
by Theorem~\ref{thm:eccentricities},  
$d(x,C_G^\whp(M)) + rad_G(M) - \whp   \leq \ecc_G^M(x)\le \ell+rad_G(M)$. Hence, $d(x,C_G^\whp(M))  \leq \whp+\ell.$ 
\end{proof}  

\medskip

We can restate the lower bound on $\ecc_G^M(x)$ in Theorem~\ref{thm:eccentricities} by using thinnes $\kappa(\calH(G))$ of metric intervals of a graph's injective hull.
\begin{lemma} \label{lem:eccUpperBoundForIntervalThinGraphs}
Let~$H$ be the injective hull of a graph~$G$. For any $M \subseteq V(G)$, every $x \in V(G)$ satisfies
$\ecc_G^M(x) \geq d_G(x,C_G^{\kappa(H)}(M)) + rad_G(M)$.
\end{lemma}
\begin{proof}
We apply the same idea as in the proof of Lemma~\ref{lem:eccLBound} with vertex~$x$, where $e_G^M(x) = e_H^M(x)= d_H(x,t)$ for some vertex $t \in M$. 
Let $h \in C_H(M)$ be a closest vertex to~$x$,
and let $y$ be a vertex on a shortest $(x,h)$-path with $d_H(h,y) = rad_G(M) - rad_H(M)$ and $e_H^M(y) = rad_G(M)$.
Then, since $G$ is isometric in $H$, there is a shortest $(x,t)$-path $P$ in $H$ such that each vertex $v \in P$ belongs to $V(G)$. Note also that, by Lemma \ref{lem:formula}, $h$ is on a shortest path from $x$ to $t$ in $H$.  
Let $y^* \in P$ such that $d_G(y^*,t) = rad_G(M)$
so that $y$ and $y^*$ belong to the same interval slice $S_{rad_G(M)}(t,x)$ in~$H$.
Therefore, $d_H(y, y^*) \leq \kappa(H)$ and, by Proposition~\ref{prop:eccentricityEqual}, $\ecc_G^M(y^*) = \ecc_H^M(y^*) \leq e_H^M(y)+ \kappa(H) = rad_G(M) + \kappa(H)$.
Hence, $\ecc_G^M(x) = d_H(x,y) + rad_G(M) =  d_H(x,y^*) + rad_G(M) \geq d_H(x,C_G^{\kappa(H)}(M)) + rad_G(M)$ as $y^*\in C_G^{\kappa(H)}(M)$.
\end{proof}

As in a $\delta$-hyperbolic graph $G$, ${\kappa({\cal H}(G))}\le 2\delta$ (see Subsection~\ref{subsec:hyperbolic}), 
%\cite{ursLang} and \cite{Chalopin_2019}), 
Lemma \ref{lem:eccUpperBoundForIntervalThinGraphs} generalizes a known result for $\delta$-hyperbolic graphs: $\ecc_G(x) \geq d_G(x,C^{2\delta}(G)) + rad(G)$ for all $x\in V(G)$~\cite{ourManuscriptHyperbolicityTerrain}. Note that similar results are known for chordal graphs: $\ecc_G(x)\ge d_G(x,C(G)) + rad(G) - 1$ for every $x\in V(G)$~\cite{Dragan2017EccentricityAT}, and for distance-hereditary graphs:  $\ecc_G(x)= d_G(x,C^1(G)) + rad(G)+1$ for every $x\in V(G)\setminus C(G)$  ~\cite{DBLP:journals/corr/abs-1907-05445}.   

\subsection{A vertex furthest from some other vertex}

Recall that in a Helly graph $G$,  for every vertex $x \in V$ and every farthest vertex $y \in F(x)$, $\ecc(y) \ge 2rad(G) - diam(C(G))$ holds~\cite{FDraganPhD}.
%\todo{D- cite for M?}

\begin{theorem}\label{lem:eccWrtDiameterOfCenter}
Let~$G$ be an $\whp$-weakly-Helly graph. Then, for every $M \subseteq V(G)$ and every $x \in V(G)$, each vertex $y \in F_G^M(x)$ satisfies $$\ecc_G^M(y) \ge 2rad_G(M) - diam_G(C_G^{2\whp}(M)) - 2\whp,$$ $$\ecc_G^M(y) \ge 2rad_G(M) - diam_G(C_G^{\whp}(M)) - 4\whp.$$
%%%%% without any M
%----
%Let $G$ be an $\whp$-weakly-Helly graph, where $x \in V$ and %$y \in F(x)$. Then, the following holds:
%\setlist{nolistsep}
%\begin{itemize}[noitemsep]
%  \item[i)] $\ecc(y)  \ge 2rad(G) - diam(C^\whp(G)) - 4\whp$, and
%  \item[ii)] $\ecc(y) \ge 2rad(G) - diam(C^{2\whp}(G)) - 2\whp$.
%\end{itemize}
\end{theorem}

\begin{proof}
Let $H := \calH(G)$ and let $e_G^M(y) = d_G(y,w)$ for some $w \in M$.
As $e_H^M(\cdot)$ is unimodal, in~$H$ there is a closest to~$x$ vertex $b_h \in C_H(M)$ such that $e_G^M(x) = d_H(x,b_h) + rad_H(M)$.
Similarly, in~$H$ there is a closest to~$y$ vertex $c_h \in C_H(M)$ such that $e_G^M(y) = d_H(y,c_h) + rad_H(M)$.
Then, $d_H(x,C_H(M)) = d_H(x,b_h) = \ecc_H^M(x) - rad_H(M) \ge rad_G(M) - rad_H(M)$ and, by symmetry, $d_H(y,c_h) \ge rad_G(M) - rad_H(M)$.
Now let $c \in I(c_h,y)$ be a vertex in~$H$ such that $d_H(c_h,c) = rad_G(M) - rad_H(M)$, and also
$b \in I(b_h,x)$ be a vertex in~$H$ such that $d_H(b_h,b) = rad_G(M) - rad_H(M)$, as illustrated in Fig.~\ref{fig:proofEccWrtDiameterOfCenter}.
By Proposition~\ref{prop:radiusInH}, $e_H^M(b) \le rad_G(M) - rad_H(M) + \ecc_H^M(b_h) = rad_G(M) \le rad_H(M) + \whp$.
By symmetry, $e_H^M(c) \le rad_H(M) + \whp$,
and so both~$b$ and~$c$ belong to $C_H^\whp(M)$.
By the triangle inequality, $rad_G(M) = d_H(b,y) \le d_H(b,c) + d_H(c,y) = d_H(b,c) + d_H(w,y) - rad_G(M)$.
That is, $\ecc_G^M(y) = d_H(w,y) \ge 2rad_G(M) - d_H(b,c) \ge 2rad_G(M) - diam_H(C_H^\whp(M))$. By Corollary~\ref{cor:helly-center-to-disk}, as $H$ is Helly, 
$\ecc_G^M(y) \ge 2rad_G(M) - diam_H(C_H^\whp(M))=2rad_G(M) - diam_H(C_H(M))-2\alpha$.
Applying now Theorem~\ref{thm:diameterOfCenters} with~$\ell=\whp$, we get $diam_H(C_H^\whp(M)) \le 2\whp + diam_G(C_G^{2\whp}(M))$ and hence $\ecc_G^M(y) \ge 2rad_G(M) -  diam_G(C_G^{2\whp}(M)) - 2\whp$.
Applying Theorem~\ref{thm:diameterOfCenters} with~$\ell=0$, we get $diam_H(C_H(M)) \le 2\whp + diam_G(C_G^{\whp}(M))$ and hence $\ecc_G^M(y) \ge 2rad_G(M) -  diam_G(C_G^{\whp}(M)) - 4\whp$.
%
%
%----
%Let $H := \calH(G)$. Since $H$ is Helly, by Proposition~\ref{prop:radiusInH}, $\ecc_G(y) = \ecc_H(y) \ge 2rad(H) - diam(C(H)) \ge 2(rad(G) - \whp) - diam(C(H))$.
%Applying Theorem~\ref{thm:diameterOfCenters} with~$\ell=0$,  $diam(C(H)) \le 2\whp + diam(C^\whp(G))$.
%Hence, $\ecc(y) \ge 2rad(G) - diam(C^\whp(G)) - 4\whp$.
%
%Let $b_h \in I(y,x)$ be a vertex in~$H$ such that $d_H(y,b_h) = rad(H)$.
%Let $w \in F(y)$, where $c_h \in I(w,y)$ is a vertex of~$H$ such that $d_H(w,c_h) = rad(H)$.
%Since $H$ is Helly, $\ecc(c_h) = \ecc(b_h) = rad(H)$.
%Let $c \in I(c_h,y)$ (and $b \in I(b_h,x)$) be a vertex in~$H$ at distance $rad(G) - rad(H)$ from~$c_h$ (from~$b_h$, respectively), as illustrated in Fig.~\ref{fig:proofEccWrtDiameterOfCenter}.
%By the triangle inequality, $rad(G) = d_H(b,y) \le d_H(b,c) + d_H(c,y) = d_H(b,c) + d_H(w,y) - rad(G)$.
%That is, $\ecc(y) = d_H(w,y) \ge 2rad(G) - d_H(b,c) \ge 2rad(G) - diam(C^\whp(H))$.
%Applying Theorem~\ref{thm:diameterOfCenters} with~$\ell=\whp$,  $diam(C^\whp(H)) \le 2\whp + diam(C^{2\whp}(G))$.
%Hence, $\ecc(y) \ge 2rad(G) -  diam(C^{2\whp}(G)) - 2\whp$.
\end{proof}

It is known that for every vertex $x$, any vertex $y\in F_G(x)$ satisfies  $\ecc_G(y) \ge 2rad(G) - 3$ if $G$ is a chordal graph or a distance-hereditary graph~\cite{10.1007/3-540-58218-5_34,DBLP:conf/esa/ChepoiD94} 
and $\ecc_G(y) \ge diam(G)-2\delta\ge 2rad(G)-6\delta-1$  if $G$ is a $\delta$-hyperbolic graph\cite{Chepoi_2008,Dragan2018RevisitingRD}. Furthermore, for every chordal or distance-hereditary graph $G$, $diam_G(C(G))\leq 3$ holds~\cite{10.1016/S0012-365X(02)00630-1,yushmanovMetricGraphProperties,chepoi:center-triang} and for every $\delta$-hyperbolic graph $G$,  
$diam_G(C^{2\delta}(G))\le 8\delta+ 1$ and $diam_G(C(G))\le 4\delta+ 1$ hold~\cite{Chepoi_2008,Dragan2018RevisitingRD}.  Recall that the diameter of the center of a Helly graph cannot be bounded. Although $C(G)$ itself is Helly and is isometric in $G$ for a Helly graph $G$~\cite{FDraganPhD}, $C(G)$ can have an arbitrarily large diameter as any Helly graph is the center of some Helly graph~\cite{FDraganPhD}. Thus, 
$diam_G(C^{\whp(G)}(G))$ cannot be bounded by a function of  $\whp(G)$ (even for $\whp(G)=0$, i.e., for Helly graphs). 

\begin{center}
\includegraphics[scale=1]{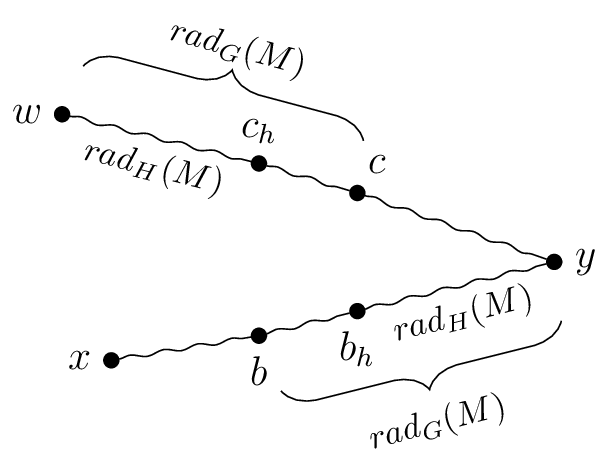}
\captionof{figure}{Illustration to the proof of Theorem~\ref{lem:eccWrtDiameterOfCenter}.}\label{fig:proofEccWrtDiameterOfCenter}
\end{center}

\subsection{Eccentricity approximating spanning tree}
%%%%%%%%%%%%%%%%%%%%%%%%%%%%%%%%%%%%%%%%%%%%%%%%%%%%%%%%%%%%
% ECC APPROX SPANNING TREES
%%%%%%%%%%%%%%%%%%%%%%%%%%%%%%%%%%%%%%%%%%%%%%%%%%%%%%%%%%%%
An \emph{eccentricity t-approximating spanning tree} is a spanning tree $T$ such that $\ecc_T(v) - \ecc_G(v) \leq t$ holds for all $v \in V(G)$.
Such a tree tries to approximately preserve only distances from $v$ to its farthest vertices, allowing for a larger increase in distances to nearby vertices.
Many classes of graphs have an eccentricity $t$-approximating spanning tree, where~$t$ is a small constant (e.g., see~\cite{nandakumar90,Prisner_2000,Dragan2017EccentricityAT,ourManuscriptHyperbolicityTerrain,Chepoi2018FastAO} and papers cited therein).
%\cite{nandakumar90,doi:10.1137/S0895480195295471,MADANLAL199697,10.1007/BFb0023484,Brandstadt1998DistanceAT,Prisner_2000,Dragan2017EccentricityAT}).
We will show that $\alpha$-weakly-Helly graphs have an eccentricity $t$-approximating spanning tree where~$t$ depends linearly on $\alpha$ and on the diameter of $C^\whp(G)$.

%Nandakumar and Parthasarasthy~\cite{nandakumar90} considered eccentricity-preserving spanning trees (i.e., $t=0$), which existed if and only if:
%(a) either $diam(G) = 2rad(G)$ and $|C(G)| = 1$, or $diam(G) = 2rad(G) - 1$ and $|C(G)| = 2$ where $C(G)$ is two adjacent vertices; and
%(b) every vertex $u \in V(G) \setminus C(G)$ has a neighbor $v$ such that $\ecc_G(v) < \ecc_G(u)$.
%
%Every additive tree $t$-spanner is clearly eccentricity $t$-approximating.
%Thus, interval graphs and distance-hereditary graphs have eccentricity 2-approximating spanning trees~\cite{doi:10.1137/S0895480195295471,MADANLAL199697,10.1007/BFb0023484}, and AT-free graphs~\cite{doi:10.1137/S0895480195295471}, strongly chordal graphs~\cite{Brandstadt1998DistanceAT}, and dually chordal graphs~\cite{Brandstadt1998DistanceAT} have 3-approximating spanning trees.
%Notably, for every $t$ there is a chordal graph without a tree $t$-spanner~\cite{doi:10.1137/S0895480195295471,10.1007/BFb0023484}, however Prisner~\cite{Prisner_2000} demonstrated that every chordal graph has an eccentricity 2-approximating spanning tree.
%Dragan et al.~\cite{Dragan2017EccentricityAT} extend that result to a much larger class of graphs, demonstrating that every~$(\alpha_1, \Delta)$-metric graph has an eccentricity 2-approximating spanning tree.

%

\begin{lemma}\label{lem:existsVertexGoodForTree}
Let $H$ be the injective hull of an $\whp$-weakly-Helly graph~$G$.
For every $M \subseteq V(G)$, there is a real vertex $c^* \in V(G)$ such that, for every vertex $x \in V(G)$, $d_G(x,c^*) + \ecc_G^M(c^*) - \ecc_G^M(x) \leq \lceil \frac{diam_H(C_H(M))}{2} \rceil + 2\whp$ holds.
\end{lemma}
\begin{proof}
%We claim there is a vertex $c^* \in V(G)$ such that $\ecc_G(x) \geq d_G(x,c^*) + \ecc_G(c^*) - (\lceil diam(C(H)) / 2 \rceil + 2\whp)$ for any $x \in V$.
Let $c \in V(H)$ be a vertex belonging to $C_H(C_H(M))$.
By Theorem~\ref{thm:closeRealVertex}, there is a real vertex $c^* \in V(G)$ such that $d_H(c,c^*) \le \whp$.

We first establish a few useful inequalities.
Since $H$ is Helly, by Lemma~\ref{lem:allCentersHellyIsom}, $C_H(M)$ and $C_H(C_H(M))$ are also Helly and isometric in~$H$.
Thus, by Lemma~\ref{lem:mDiamRadInHelly}, $rad_H(M) = \lceil diam_H(M) / 2 \rceil$ and $rad_H(C_H(M)) = \lceil diam_H(C_H(M)) / 2 \rceil$.
Let $c_x \in C_H(M)$ be a closest vertex to $x$.
By Lemma~\ref{lem:formula}, $\ecc_G^M(x) = \ecc_H^M(x) =  d_H(x,C_H(M)) + rad_H(M) = d_H(x,c_x) + \lceil diam_H(M) / 2 \rceil$.
By %Lemma~\ref{prop:eccentricityEqual} and 
the triangle inequality, $\ecc_G^M(c^*) = \ecc_H^M(c^*) \le \ecc_H^M(c) + \whp = \lceil diam_H(M) / 2 \rceil + \whp$.
Since $G$ is isometric in~$H$, $d_G(x,c^*) = d_H(x,c^*) \le d_H(x,c_x) + d_H(c_x,c) + d_H(c,c^*) \le d_H(x,c_x) + rad_H(C_H(M)) + \whp = d_H(x,c_x) + \lceil diam_H(C_H(M)) / 2 \rceil + \whp$.
We are now ready to prove the claim.
\begin{align*}
    \ecc_G^M(x) = \ecc_H^M(x)
    &= d_H(x,c_x) + \lceil diam_H(M) / 2 \rceil \\
    &\ge (d_G(x,c^*) - \lceil diam_H(C_H(M)) / 2 \rceil - \whp) + \lceil diam_H(M) / 2 \rceil \\
    %&= d_G(x,c^*) - \lceil diam_H(C_H(M)) / 2 \rceil + (\lceil diam_H(M) / 2 \rceil - \whp) \\
    &\ge d_G(x,c^*) - \lceil diam_H(C_H(M)) / 2 \rceil + \ecc_G^M(c^*) - 2\whp.
\end{align*}

Therefore, $d_G(x,c^*) + \ecc_G^M(c^*) - \ecc_G^M(x) \le %d_G(x,c^*) + \ecc_G^M(c^*) - (d_G(x,c^*) - \lceil diam_H(C_H(M)) / 2 \rceil + \ecc_G^M(c^*) - 2\whp) \le \lceil %
diam_H(C_H(M)) / 2 \rceil + 2\whp$.
\end{proof}

\begin{theorem}\label{thm:spanningTree}
Let $G$ be an $\whp$-weakly-Helly graph.
For every $M \subseteq V(G)$, $G$ has a spanning tree $T$ such that, for all $x \in V(G)$, $\ecc_T^M(x) - \ecc_G^M(x) \le \lceil diam_G(C_G^\whp(M)) / 2 \rceil + 3\whp$.
\end{theorem}
\begin{proof}
By Lemma~\ref{lem:existsVertexGoodForTree}, there exists a vertex~$c^* \in V(G)$ such that any vertex~$x \in V(G)$ satisfies $d_G(x,c^*) + \ecc_G^M(c^*) - \ecc_G^M(x) \leq \lceil diam_H(C_H(M)) / 2 \rceil + 2\whp$.
Let $T$ be a BFS tree of $G$ rooted at~$c^*$.
Then, any vertex~$x \in V(G)$ has $\ecc_T^M(x) - \ecc_G^M(x) \le d_T(x,c^*) + \ecc_T^M(c^*) - \ecc_G^M(x)$.
As distances to~$c^*$ are preserved in~$T$, one obtains
$\ecc_T^M(x) - \ecc_G^M(x) \le \lceil diam_H(C_H(M)) / 2 \rceil + 2\whp$.
Applying Theorem~\ref{thm:diameterOfCenters} with $\ell = 0$, $diam_H(C_H(M)) \le diam_G(C_G^\whp(M)) + 2\whp$.
Hence, $\ecc_T^M(x) - \ecc_G^M(x)  \le \lceil (diam_G(C_G^\whp(M)) + 2\whp)/2) \rceil + 2\whp$, establishing the desired result.
\end{proof}

It is known that every chordal graph and every distance-hereditary graph has an eccentricity 2-approximating spanning tree~\cite{Prisner_2000,Dragan2017EccentricityAT,DBLP:journals/ipl/Dragan20,doi:10.1137/S0895480195295471} and  every $\delta$-hyperbolic graph has an eccentricity  $(4\delta+ 1)$-approximating spanning  tree~\cite{ourManuscriptHyperbolicityTerrain}. These graph classes have additional nice properties that allowed to get for them sharper results than the one of Theorem \ref{thm:spanningTree} which holds for general $\whp$-weakly-Helly graphs.

\subsection{Eccentricity terrain}

%\subsection{Eccentricity terrain of $\whp$-weakly-Helly graph}
%%%%%%%%%%%%%%%%%%%%%%%%%%%%%%%%%%%%%%%%%%%%%%%%%%%%%%%%%%%%
% TERRAIN
%%%%%%%%%%%%%%%%%%%%%%%%%%%%%%%%%%%%%%%%%%%%%%%%%%%%%%%%%%%%
We can describe the behavior of the eccentricity function in a graph in terms of graph's \emph{eccentricity terrain}, that is, how vertex eccentricities change along vertices of a shortest path to $C_G^{\whp}(M)$. 
We define an ordered pair of vertices $(u,v)$, where $(u,v) \in E$,
as an \emph{up-edge} if $\ecc_G^M(u) < \ecc_G^M(v)$,
as a \emph{down-edge} if $\ecc_G^M(u) > \ecc_G^M(v)$, and
as a \emph{horizontal-edge} if $\ecc_G^M(u) = \ecc_G^M(v)$.
Let $\up(P(y,x))$, $\horiz(P(y,x))$, and $\down(P(y,x))$ denote the number of up-edges, horizontal-edges, and down-edges along a shortest path $P(y,x)$ from $y \in V(G)$ to $x \in V(G)$.

We note that Lemma~\ref{lem:upHorizontalEdgesBound} is shown in~\cite{ourManuscriptHyperbolicityTerrain} for the eccentricity function $\ecc_G(\cdot)$. For completeness, we provide a proof that it holds for $\ecc_G^M(\cdot)$ for any $M \subseteq V(G)$.

\begin{lemma}\label{lem:upHorizontalEdgesBound}~\cite{ourManuscriptHyperbolicityTerrain} Let $G$ be an arbitrary graph and $M \subseteq V(G)$. For any shortest path $P(y,x)$ of $G$ from a vertex $y$ to a vertex $x$ the following holds:
\setlist{nolistsep}
\begin{itemize}[noitemsep]
  \item[i)] $\down(P(y,x)) - \up(P(y,x)) = \ecc_G^M(y) - \ecc_G^M(x)$, and
  \item[ii)] $2\up(P(y,x)) + \horiz(P(y,x)) = d_G(y,x) - (\ecc_G^M(y) - \ecc_G^M(x))$.
\end{itemize}
\end{lemma}
\begin{proof}
We use an induction on $d_G(y,v)$ for any vertex $v \in P(y,x)$.
First, assume that $v$ is adjacent to $y$.
If~$(y,v)$ is an up-edge, then $\ecc_G^M(y) - \ecc_G^M(v) = -1$ and $\down(P(y,v)) - \up(P(y,v)) = -1$.
If~$(y,v)$ is a horizontal-edge, then $\ecc_G^M(y) - \ecc_G^M(v) = 0$ and $\down(P(y,v)) - \up(P(y,v)) = 0$.
If~$(y,v)$ is a down-edge, then $\ecc_G^M(y) - \ecc_G^M(v) = 1$ and $\down(P(y,v)) - \up(P(y,v)) = 1$.
Now consider an arbitrary vertex $v \in P(y,x)$ and assume, by induction, that $\ecc_G^M(y) - \ecc_G^M(v) = \down(P(y,v)) - \up(P(y,v))$.
Let vertex $u \in P(y,x)$ be adjacent to $v$ with $d(y,u) = d(y,v) + 1$.
By definition, $\down(P(y,u)) = \down(P(y,v)) + \down((v,u))$ and
$\up(P(y,u)) = \up(P(y,v)) + \up((v,u))$.
We consider three cases based on the classification of edge $(v,u)$.

If~$(v,u)$ is an up-edge, then $\ecc_G^M(u) = \ecc_G^M(v) + 1$, $\up((v,u)) = 1$, and $\down((v,u)) =0=  \up((v,u)) - 1$.
By the inductive hypothesis,
$\down(P(y,u)) = \down(P(y,v)) + \down((v,u))
                       = \up(P(y,v)) + \ecc_G^M(y) - \ecc_G^M(v) + \up((v,u)) - 1
                       = \up(P(y,u)) + \ecc_G^M(y) - \ecc_G^M(v) - 1
                       = \up(P(y,u)) + \ecc_G^M(y) - \ecc_G^M(u)$.

If~$(v,u)$ is a horizontal-edge, then $\ecc_G^M(u) = \ecc_G^M(v)$, $\up((v,u)) = 0=\down((v,u))$.
By the inductive hypothesis,
$\down(P(y,u)) = \down(P(y,v)) + \down((v,u))
                       = \up(P(y,v)) + \ecc_G^M(y) - \ecc_G^M(v) + \up((v,u))
                       = \up(P(y,u)) + \ecc_G^M(y) - \ecc_G^M(v)
                       = \up(P(y,u)) + \ecc_G^M(y) - \ecc_G^M(u)$.

If~$(v,u)$ is a down-edge, then $\ecc_G^M(u) = \ecc_G^M(v) - 1$, $\up((v,u)) = 0$, and  $\down((v,u)) = \up((v,u)) + 1$.
By the inductive hypothesis,
$\down(P(y,u)) = \down(P(y,v)) + \down((v,u))
                       = \up(P(y,v)) + \ecc_G^M(y) - \ecc_G^M(v) + \up((v,u)) + 1
                       = \up(P(y,u)) + \ecc_G^M(y) - \ecc_G^M(v) + 1
                       = \up(P(y,u)) + \ecc_G^M(y) - \ecc_G^M(u)$.

This completes the proof of (i) that $\down(P(y,x)) - \up(P(y,x)) = \ecc_G^M(y) - \ecc_G^M(x)$.
To complete the proof of (ii), it now suffices to see that
$\up(P(y,x)) + \horiz(P(y,x)) + (\up(P(y,x)) + \ecc_G^M(y) - \ecc_G^M(x)) = \up(P(y,x)) + \horiz(P(y,x)) + \down(P(y,x)) = d_G(y,x)$.
Hence,  $2\up(P(y,x)) + \horiz(P(y,x)) = d_G(y,x) - (\ecc_G^M(y) - \ecc_G^M(x))$.
\end{proof}

\begin{theorem}\label{thm:upHorizontalEdgesBoundWH}
Let $G$ be an $\whp$-weakly-Helly graph and $M \subseteq V(G)$. For any shortest path $P(y,x)$ of $G$ from a vertex $y \notin C_G^\whp(M)$ to a closest vertex $x \in C_G^\whp(M)$, $2\up(P(y,x)) + \horiz(P(y,x)) \leq 2\whp$ holds.
\end{theorem}
\begin{proof}
Let $x \in C_G^\whp(M)$ be closest to $y$. Hence, $\ecc_G^M(x) = rad_G(M) + \whp$.
Applying Theorem~\ref{thm:eccentricities}, $\ecc_G^M(y) \geq d_G(y,C_G^\whp(M)) + rad_G(M) - \whp$.
Hence, $d_G(y,x) = d_G(y,C_G^\whp(M)) \le \ecc_G^M(y) - rad_G(M) + \whp = \ecc_G^M(y) - \ecc_G^M(x) + 2\whp$.
Applying Lemma~\ref{lem:upHorizontalEdgesBound}, we  obtain the desired result.
\end{proof}

As a consequence of Theorem~\ref{thm:upHorizontalEdgesBoundWH},
at most $2\whp(G)$ non-descending edges can occur along every shortest path from any vertex~$y \notin C^{\whp}(G)$ to $C^{\whp}(G)$. Hence, in any shortest path to $C^{\whp}(G)$, the number of vertices with locality more than 1 does not exceed $2\whp(G)$. These kind of results were only known for chordal graphs and distance-hereditary graphs (at most one non-descending edge, i.e., horizontal-edge~\cite{Dragan2017EccentricityAT,DBLP:journals/corr/abs-1907-05445})  and for $\delta$-hyperbolic graphs (at most $4\delta$ non-descending edges~\cite{ourManuscriptHyperbolicityTerrain}). 
%%%%%%

%%%%%%%%%%%%%%%%%%%%%%%%%%%%%%%%%%%%%%%%%%%%%%%%%%%%%%%%%%%%
%%%%%%%%%%%%%%%%%%%%%%%%%%%%%%%%%%%%%%%%%%%%%%%%%%%%%%%%%%%%
% CLASSES THAT ARE k-WEAKLY-HELLY
%%%%%%%%%%%%%%%%%%%%%%%%%%%%%%%%%%%%%%%%%%%%%%%%%%%%%%%%%%%%
%%%%%%%%%%%%%%%%%%%%%%%%%%%%%%%%%%%%%%%%%%%%%%%%%%%%%%%%%%%%
\section{Classes of graphs having small \hg{}}\label{sec:cases}
Let $G$ be a graph and let $\alpha:=\alpha(G)$ be its Helly-gap.  
We now focus on the upper bound on $\whp$ for some special graph classes and identify those classes which exhibit a topological Helly-likeness, that is, classes which have small $\alpha$.
Recall that, by definition, Helly graphs are 0-weakly-Helly.
It was also recently shown that each Helly vertex in the injective hull of a distance-hereditary graph is adjacent to a real vertex~\cite{ourManuscriptInjectiveHulls};
thus distance-hereditary graphs are 1-weakly-Helly (this result follows also from an earlier result on existence of $r$-dominating cliques in distance-hereditary graphs\cite{10.1007/3-540-58218-5_34}). {\em  Distance-hereditary graphs} can be defined as the graphs where each induced path is a shortest path.

%\todo{Helly-gap for euclidean plane, some polygons and metrics}

We consider several other graph classes and summarize our findings in Table~\ref{table:kWeakHellyClasses}.

\begin{table}[h]
\begin{center}
\begin{tabular}{ |p{5cm}|p{3cm}|  } 
 \hline
 \textbf{Graph class}     & \textbf{\hg{}} \\         \hline
 Helly                    & $\whp(G)=0$ \\                             \hline
 Distance-hereditary      & $\whp(G) \leq 1$ \\                        \hline
 $k$-Chordal              & $\whp(G) \leq \lfloor k/2 \rfloor$ \\      \hline
 Chordal                  & $\whp(G) \leq 1$ \\                        \hline
 AT-free                  & $\whp(G) \leq 2$ \\                        \hline
 $n \times n$ Rectilinear grid & $\whp(G) = 1$ \\                           \hline
 $\delta$-Hyperbolic      & $\whp(G) \leq 2\delta $ \\                 \hline
 Cycle $C_n$              & $\whp(G) = \lfloor n/4 \rfloor$ \\         \hline
 $\alpha_i$-Metric        & $\whp(G) \le \lceil i/2 \rceil$ \\        \hline
 Tree-breadth $tb(G)$     & $\whp(G) \leq tb(G)$  \\            \hline
 Tree-length $tl(G)$      & $\whp(G) \leq tl(G)$  \\          \hline
 Tree-width $tw(G)$       & no relation  \\          \hline
 Bridged                  & unbounded\\           \hline
\end{tabular}
\end{center}
\caption{The bound on the \hg{} $\whp(G)$ for various graph classes.}\label{table:kWeakHellyClasses}
\end{table}

\subsection{Rectilinear grids}\label{subsec:grids} Let $G$ be an $n \times n$ rectilinear grid (the cartesian product of two paths of length $n$). $G$ can be isometrically embedded into a $2n \times 2n$ king grid $H$ (the strong product of two paths of length $2n$, and a natural subclass of Helly graphs) wherein the four extreme vertices of $G$ are the midpoints of each side of $H$ (see Fig.~\ref{fig:planrGridInKingGrid}). One obtains $\calH(G)$ by removing from $H$ any peripheral vertices that are not present in $G$. Therefore, each Helly vertex $h$ of $\calH(G)$ corresponds to one of the $n^2$ squares of the grid $G$, where $h$ is adjacent to the four corners of the square and to each of the Helly vertices corresponding to the 2-4 adjacent squares. Hence, the \hg{} of $G$ is 1.

\begin{center}
\includegraphics[scale=1]{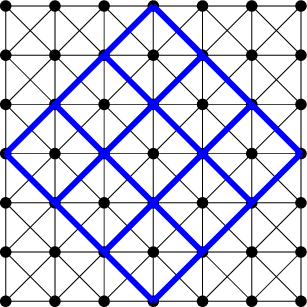}
\captionof{figure}{An $n \times n$ rectilinear grid (shown in bold blue lines) is isometrically embedded into a $2n \times 2n$ king grid (shown in thin black lines).}\label{fig:planrGridInKingGrid}
\end{center}

\subsection{Graphs of bounded %$\delta$-
hyperbolicity}\label{subsec:hyperbolic}
Gromov~\cite{Gromov1987} defines $\delta$-hyperbolic graphs via a simple 4-point condition:
for any four vertices $u,v,w,x$, the two larger of the three distance sums $d(u,v) + d(w,x)$, $d(u,w) + d(v,x)$, and $d(u,x) + d(v,w)$ differ by at most $2\delta \geq 0$.
The smallest value $\delta$ for which $G$ is $\delta$-hyperbolic is called the \emph{hyperbolicity} of $G$ and is denoted by $\delta(G)$. As we have mentioned earlier, we omit~$G$ from subindices when the graph is known by context.

Recall that $\kappa(G)$ denotes the smallest value $\kappa$ for which all intervals of~$G$ are $\kappa$-thin.
% define slices
% define interval
% define tau
% define hyperbolicity
The following lemma is a folklore and easy to show using the definition of hyperbolicity.
\begin{lemma} \label{half-thin}
 For any graph $G$, $\kappa(G) \leq {2}\delta(G)$.
\end{lemma}
\begin{proof}
    Consider any interval $I(x,y)$ in $G$ and arbitrary two vertices $u,v\in S_k(x,y)$. Consider the three distance sums
    $S_1=d(x,y)+d(u,v)$, $S_2=d(x,u)+d(y,v)$, $S_3=d(x,v)+d(y,u)$. As $u,v\in S_k(x,y)$, we have $S_2=S_3=d(x,y)\leq S_1$. Hence,
    $2 \delta(G)\ge S_1-S_2= d(x,y)+d(u,v)-d(x,y)=d(u,v)$ for any two vertices from the same slice of $G$, i.e., $2 \delta(G)\ge \kappa(G)$.
\end{proof}

We now show that the \hg{} of any graph~$G$ is at most $\kappa(\calH(G))$.

\begin{lemma} \label{lem:hDeficitAndThinness}
For any vertex $h \in V(\calH(G))$ there is a real vertex $v \in V(G)$ such that $d(h,v) \leq \kappa(\calH(G))$.
\end{lemma}
\begin{proof}
Let $H := \calH(G)$ and consider a vertex $h \in V(H)$.
By Proposition~\ref{prop:shortestPathSubsetOfReal}, $h$ belongs to a shortest path $P_H(x,y)$ where $x,y \in V(G)$.
Let $d_H(x,h)=\ell$.
Because $G$ is isometric in $H$, there is a shortest path $P_G(x,y):= (v_0, ...., v_k)$, where $x=v_0$, $y=v_k$ and all $v_i$ in $P_G$ are real vertices.
By assumption, $d_H(h, v_\ell) \leq \kappa(H)$.
\end{proof}

\begin{corollary}
Any graph~$G$ is $\whp$-weakly-Helly for $\whp \le \kappa(\calH(G))$.
\end{corollary}
\begin{proof}
The result follows from Theorem~\ref{thm:closeRealVertex} and Lemma~\ref{lem:hDeficitAndThinness}.
\end{proof}
Note that an $n \times n$ rectilinear grid gives an example of a graph where $\alpha=1$ and $\kappa(\calH(G))=2n$. 

Urs Lang~\cite{ursLang} has established that the $\delta$-hyperbolicity of a graph $G$ is preserved in $H=\calH(G)$. % and (b) that any Helly vertex of $\calH(G)$ is within $2\delta$ from a real vertex of $G$.
As hyperbolicity is preserved in $\calH(G)$,  by Lemma~\ref{lem:hDeficitAndThinness} and Lemma~\ref{half-thin},
for any Helly vertex $h \in V(H)$ there is a real vertex $v \in V(G)$ such that $d_H(h,v) \leq \kappa(\calH(G)) \le 2\delta(\calH(G)) = 2\delta(G)$.
Hence, $\delta$-hyperbolic graphs are $2\delta$-weakly-Helly. The latter result follows also from \cite{DBLP:conf/approx/ChepoiE07}.

\begin{remark}
Several graph classes are $\delta$-hyperbolic for a bounded value $\delta$, including $k$-chordal graphs and graphs without an asteroidal triple (AT). 
Three pairwise non-adjacent vertices of a graph form an AT if every two of them are connected by a path that avoids the neighborhood of the third.
A graph is $k$-chordal provided it has no induced cycle of length greater than $k$.
It is known~\cite{Wu2011} that $k$-chordal graphs have hyperbolicity at most $\frac{\lfloor k/2 \rfloor}{2}$.
Therefore, $k$-chordal graphs are $\lfloor k/2 \rfloor$-weakly-Helly.
Chordal graphs, which are exactly the 3-chordal graphs, are 1-weakly-Helly  (this result follows also from an earlier result on existence of $r$-dominating cliques in chordal  graphs\cite{DBLP:journals/dm/DraganB96}).
The 2-weakly-Helly graphs also include all 5-chordal graphs such as AT-free graphs,
all 4-chordal graphs such as weakly-chordal and cocomparability graphs, % and distance-hereditary graphs, but we can do better for them.
as well as all 1-hyperbolic graphs including interval graphs and permutation graphs. The definitions of graph classes not given here can be found in \cite{doi:10.1137/1.9780898719796}. 
\end{remark}

\begin{remark}
Observe that $\whp(G)$ can be arbitrarily smaller than the $\delta$-hyperbolicity of a graph. Consider a $(2r \times 2r)$ king grid $G$, i.e., the strong product of two paths each of even length~$2r$.
King grids form a natural subclass of Helly graphs, and therefore $\calH(G) = G$.
Thus, %for any vertex $h \in V(\calH(G))$, then also $h \in V(G)$, so the distance from $h$ to a real vertex is 0.
%By Theorem~\ref{thm:closeRealVertex}, 
$G$ is 0-weakly-Helly, whereas $\delta(G)=r$.
\end{remark}

\subsection{Cycles}\label{subsec:cycles}
Recall that $C_n$ denotes a cycle of length~$n$.

\begin{lemma}
The \hg{} of a cycle $C_n$ is $\lfloor n/4 \rfloor$.
\end{lemma}
\begin{proof}
Let $G$ be a cycle of length~$n$.
Clearly, $diam(G) = \lfloor n/2 \rfloor = rad(G)$.
By Proposition~\ref{cor:whpBoundedByDiam}, $G$ is $\whp$-weakly-Helly for $\whp \le \lfloor diam(G)/2 \rfloor = \lfloor n/4 \rfloor$.
By Lemma~\ref{lem:whpLowerBoundByDiam2Rad}, $\whp \geq \lfloor (2rad(G) - diam(G))/2 \rfloor = \lfloor n /4 \rfloor$.
\end{proof}

It should also be noted that for a self-centered graph (i.e., a graph where all vertex eccentricities are equal), from $\lfloor (2rad(G) - diam(G)) / 2 \rfloor \le \whp(G) \le \lfloor diam(G) / 2 \rfloor$ and since for a self-centered graph $G$,  $rad(G) = diam(G)$, we get $\whp(G) = \lfloor diam(G) / 2 \rfloor=\lfloor rad(G) / 2 \rfloor$. Cycles are self-centered graphs.   

%\begin{lemma} \label{lem:deltaWeaklyHellyHyperbolic}
%If $G$ is a $\delta$-hyperbolic graph, then $G$ is $2\delta$-weakly-Helly.
%\end{lemma}
%\begin{proof}
%Recall that hyperbolic graphs are closed under Hellification by Lemma~\ref{prop:hyperbolicityEqual}.
%Thus, by Lemma~\ref{lem:hDeficitAndThinness} and Lemma~\ref{half-thin}, we know that
%for any Helly vertex $h \in V(H)$ and real vertex $x \in V(G)$, we have $d_H(z,h) \leq 2\delta$.
%By Theorem~\ref{thm:closeRealVertex}, $G$ is $2\delta$-weakly-Helly.
%\end{proof}

\subsection{Graphs with an $\alpha_i$-metric}\label{subsec:alphaimetric}
Introduced by Chepoi and Yushmanov~\cite{chepoi:center-triang,yushmanovMetricGraphProperties}, a graph~$G$ is said to have an $\alpha_i$-metric if it satisfies the following: for any $x,y,z,v \in V(G)$ such that $z \in I(x,y)$ and $y \in I(z,v)$, $d_G(x,v) \ge d_G(x,y) + d_G(y,v) - i$ holds. 
For every graph $G$ with an $\alpha_i$-metric, $diam(G)\ge 2rad(G)-i-1$ holds~\cite{yushmanovMetricGraphProperties}. Several graph classes have an $\alpha_i$-metric~\cite{yushmanovMetricGraphProperties}. Ptolemaic graphs are exactly the graphs with $\alpha_0$-metric~\cite{yushmanovMetricGraphProperties}. Chordal graphs are a subclass of graphs with $\alpha_1$-metric~\cite{chepoi:center-triang}. All graphs with $\alpha_1$-metric are characterized in~\cite{yushmanovMetricGraphProperties}.
In a private communication, discussing the results of this paper,  Victor Chepoi asked if graphs with an $\alpha_i$-metric are $\whp$-weakly-Helly for an $\whp$ depending only on $i$. 
The following lemma answers this question in the affirmative. 
\begin{lemma}
Any graph $G$ with an $\alpha_i$-metric is $\whp$-weakly-Helly for $\whp \le \lceil i/2 \rceil$.
\end{lemma}
\begin{proof}
We prove by induction on the number $k$ of disks in a family of pairwise intersecting disks $\calF = \{D(v,r(v)) : v \in M \subseteq V(G)\}$.
Let $y \in M$ and pick a vertex $c$ which belongs to $\bigcap_{v \in M \setminus \{y\}}D(v,r(v)+\whp)$ such that $c$ is closest to $y$.
If $d(c,y) \le r(y) + \whp$, then $c$ is common to all disks of $\calF$ inflated by $\whp$ and we are done.
Assume now that $d(c,y) > r(y) + \whp$.
Let $c' \in S_1(c,y)$.
By choice of $c$, there is a disk $D(x,r(x)) \in \calF$ such that $d(c,x) = r(x) + \whp = d(c',x) - 1$.
Hence, $c \in I(x,c')$ and $c' \in I(c,y)$. 
Applying $\alpha_i$-metric to $x,c,c',y$, one obtains $d(x,y) \ge d(x,c) + d(c,y) - i > (r(x) + \whp) + (r(y) + \whp) - i=r(x) + r(y) + 2\whp -i$. If now $i\le 2\whp$, we get $d(x,y)> r(x) + r(y)$, contradicting the fact that disks $D(x,r(x))$ and $D(y,r(y))$ intersect. Thus, for 
every $i$ (even or odd) such that $i\le 2\whp$, $d(c,y) \le r(y) + \whp$ must hold. That is, $G$ is $\whp$-weakly-Helly for $\whp = \lceil i/2 \rceil$ (pick $\whp=i/2$, when $i$ is even,  and $\whp=(i+1)/2$, when $i$ is odd).
%Thus, $i > r(x) + r(y) + 2\whp - d(x,y)$, and as $r(x) + r(y) \ge d(x,y)$, one obtains $\whp < i/2$.
%Since $\whp$ is an integer, $\whp \le \lfloor i/2 \rfloor$.
\end{proof}

\begin{remark}
Observe that a graph $G$ with an $\alpha_i$-metric can have Helly-gap $\whp(G)$ that is arbitrarily smaller than $\lceil i/2 \rceil$. Consider a $(1 \times \ell)$ rectilinear grid $G$. Let $(x,y)$ and $(z,v)$ be edges on extreme ends of $G$ so that $d(x,z) = d(y,v) = \ell$ and $d(y,z) = d(x,v) = \ell+1$. Then, $z \in I(x,v)$, $v \in I(z,y)$, and $1 = d(x,y) \ge d(x,v) + d(v,y) - i = 2\ell +1 - i$.
Thus, $G$ has an $\alpha_i$-metric for $i \ge 2\ell$ whereas $\whp(G) = 1$.
\end{remark}

\subsection{Graphs of bounded tree-breadth, tree-length, or tree-width}\label{subsec:treeDecompisitionParameters}
A \emph{tree-decomposition} $(\calT,T)$~\cite{ROBERTSON1986309} for a graph $G=(V,E)$ is a family $\calT = \{ B_1, B_2, ... \}$ of subsets of $V$, called \emph{bags}, such that $\calT$ forms a tree $T$ with the bags in $\calT$ as nodes which satisfy the following conditions:
(i) each vertex is contained in a bag,
(ii) for each edge $(u,v) \in E$, $\calT$ has a bag $B$ with $u,v \in B$, and
(iii) for each vertex $v \in V$, the bags containing $v$ induce a subtree of $T$.
The width of a tree decomposition is the size of its largest bag minus one.
A tree decomposition has breadth $\rho$ if, for each bag $B$, there is a vertex $v$ in $G$ such that $B \subseteq D_G(v,\rho)$.
A tree decomposition has length $\lambda$ if the diameter in~$G$ of each bag $B$ is at most $\lambda$.
The \emph{tree-width} $tw(G)$~\cite{ROBERTSON1986309}, \emph{tree-breadth} $tb(G)$~\cite{DBLP:journals/algorithmica/DraganK14} and  \emph{tree-length} $tl(G)$~\cite{DOURISBOURE20072008} are the minimum width, breadth, and length, respectively, among all possible tree decompositions of $G$.
By definition, $tb(G) \le tl(G) \le 2tb(G)$, as for any graph $G$ and any set $M\subseteq V(G)$, $rad_G(M)\le diam_G(M)\le 2rad_G(M)$ holds.

\begin{lemma} Any graph $G$ is $\alpha$-weakly-Helly for $\alpha \le tb(G)$ and $\alpha \le tl(G)$.
\end{lemma}
\begin{proof}
Let $G$ have tree-breadth $tb(G) \le \rho$.
Consider a corresponding  tree-decomposition $(\calT,T)$ of $G$ of breadth $\rho$, and let $\calF = \{ D_G(v, r(v)) : v \in S \subseteq V(G)\}$ be a family of disks of $G$ that pairwise intersect.
Observe that bags containing vertices of a disk induce a subtree of $T$ and the subtrees of $T$ corresponding to disks of $\calF$ pairwise intersect in $T$. 
If the subtrees of a tree pairwise intersect, then they have a common node in $T$.
Therefore, there is a bag $B \in \calT$ such that each disk in $\calF$ intersects $B$ in $G$.
Let vertex $w$ be the center of bag $B$, i.e., $B \subseteq D_G(w, \rho)$. 
So, each $v \in S$ satisfies $w \in D_G(v, r(v) + \rho)$.
Hence, $G$ is $\alpha$-weakly-Helly where $\alpha \le tb(G) \le tl(G)$.
\end{proof}

\begin{remark}
While $\alpha:=\alpha(G)$ is upper bounded by tree-breadth and tree-length, $\alpha$ can be arbitrarily far from tree-width and arbitrarily smaller than tree-length.
Consider the $(r \times r)$ rectilinear grid $G$ which is 1-weakly-Helly but $tw(G)=r$ (cf.~\cite{ROBERTSON198449}) and $tl(G) = 2r$ (cf.~\cite{DOURISBOURE20072008,Adcock_2016}).
On the other hand, let $G$ be a cycle of size $4k$;
the \hg{} of $G$ is $k$, yet $tw(G) = 2$.
\end{remark}

\subsection{Bridged graphs}\label{subsec:bridged}
A graph is \emph{bridged}\cite{FARBER1987249} if it contains no isometric cycles of length greater than 3.
As any isometric subgraph is an induced subgraph, bridged graphs form a superclass of chordal graphs.
Interestingly, we find that although chordal graphs are 1-weakly-Helly, the \hg{} of a bridged graph can be arbitrarily large.
Consider the bridged graph $G$ in Fig.~\ref{fig:bridgedUnbounded} with each side of length $s=4k$ for some integer~$k$.
Clearly the disks centered at $x,y,z$ with radius $2k$ pairwise intersect and have no common intersection.
However, only extending all radii by $k$ yields middle vertex $u$ as a common intersection.
Thus, the \hg{} of $G$ is at least $k$, where $k$ can be arbitrarily large.

\begin{center}
\includegraphics[scale=1]{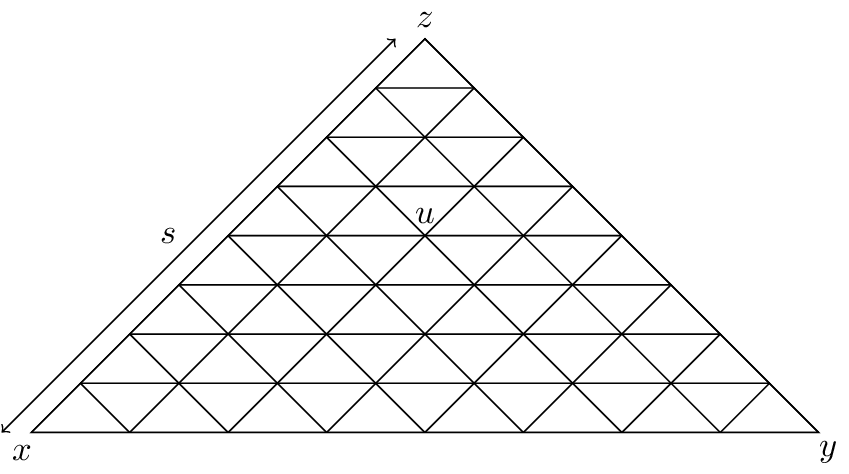}
\captionof{figure}{Example that bridged graphs have arbitrarily large \hg{}.}\label{fig:bridgedUnbounded}
\end{center}

\section{Conclusion}\label{sec:conclusion}
We introduced a new metric parameter for a graph, the Helly-gap $\alpha(G)$. We characterized $\whp$-weakly-Helly graphs with respect to distances in their injective hulls and used this as a tool to generalize many eccentricity related results known for Helly graphs to  much larger
family of graphs, $\whp$-weakly-Helly
graphs. In particular, we related the diameter, radius, center, and all eccentricities in~$G$ to their counterparts in~$\calH(G)$. We provided estimates on $\whp(G)$ and on eccentricities based on a variety of conditions, thereby generalizing some  eccentricity related results known not only for Helly graphs but also for distance-hereditary graphs, chordal graphs, and $\delta$-hyperbolic graphs. Several additional graph classes are identified which have a bounded \hg{}. 

Two interesting questions remain open.
% (max and uniform)
First, is there a characterization of $\whp(G)$ using only the radius and diameter of every~$M \subseteq V(G)$? Is it true that $\whp(G)=\max_{M \subseteq V(G)}\lfloor (2rad_G(M) - diam_G(M)) / 2 \rfloor$?
As it was recently shown \cite{newDDG2020}, the equality $\max_{M \subseteq V(G)}\lfloor (2rad_G(M) - diam_G(M)) / 2 \rfloor=0$ characterizes the Helly graphs (i.e., the graphs with $\alpha(G)=0$): $G$ is a Helly graph if and only if $\max_{M \subseteq V(G)}\lfloor (2rad_G(M) - diam_G(M)) / 2 \rfloor=0$. Second, as $\whp$-weakly-Helly graphs are a far reaching superclass of $\delta$-hyperbolic graphs, under what additional conditions are $\whp$-weakly-Helly graphs reduced to $\delta$-hyperbolic graphs for some $\delta$ dependent on $\whp$?

\medskip
\noindent
{\bf Acknowledgment:} 
We are grateful to Victor Chepoi for reading an earlier version of this paper and providing a few interesting remarks. 
\bibliographystyle{abbrv}
\bibliography{bibliography}

\end{document}